\newif\iflncs\lncsfalse
\newif\ifmai\maifalse
\newif\ifanon\anonfalse
\newif\iffull\fulltrue
\iflncs
\documentclass{llncs}
\usepackage{etex}
\else
\documentclass[11pt]{article}
\usepackage{etex}
\fi
\usepackage[utf8]{inputenc}
\usepackage{amsmath}
\usepackage{bm}
\usepackage{url,amsmath,amsfonts}

\usepackage{chngcntr}
\usepackage{apptools}
\AtAppendix{\counterwithin{definition}{section}}
\AtAppendix{\counterwithin{lemma}{section}}
\AtAppendix{\counterwithin{theorem}{section}}
\AtAppendix{\counterwithin{remark}{section}}

\usepackage{algorithm,algcompatible,amssymb,amsmath}
\renewcommand{\COMMENT}[2][.5\linewidth]{%
  \leavevmode\hfill\makebox[#1][l]{//~#2}}
\algnewcommand\algorithmicto{\textbf{to}}
\algnewcommand\RETURN{\State \textbf{return} }

\usepackage{ae}
\usepackage{aecompl}
\usepackage[mathscr]{euscript}
\usepackage{graphics}
\usepackage{graphicx}

\usepackage{vwcol}
\usepackage{multicol}
\usepackage{pstricks}
\usepackage{stackengine}
\usepackage{algorithm}
\usepackage[noend]{algpseudocode}
\usepackage[mathscr]{euscript}%
\usepackage[english]{babel}
\usepackage[T1]{fontenc}
\usepackage[all]{xy}
\usepackage{txfonts}
\usepackage{algorithm,algpseudocode}
\usepackage{url}

\usepackage{txfonts}
\iflncs
\else
\usepackage{geometry}
\ifmai
\expandafter\def\expandafter\normalsize\expandafter{%
    \normalsize
    \setlength\abovedisplayskip{1pt}
    \setlength\belowdisplayskip{1pt}

    \setlength\abovedisplayshortskip{1.1pt}
    \setlength\belowdisplayshortskip{1.1pt}

}
\fi
\fi

\usepackage{amssymb,amsmath}
\usepackage{wasysym}
\usepackage{pgf}
\usepackage{tikz}
\usepackage{wrapfig}
\usepackage{epstopdf}

\newenvironment{proof}{{\sc Proof}~}{ \hfill $\Box$ \vskip .2mm }
\newenvironment{proof_of}[1]{{\sc Proof of #1}~}{ \hfill $\Box$ \vskip .2mm}
\newenvironment{notation}{\emph{Notation:~}}{}

\iflncs{}\else \newtheorem{definition}{Definition}\fi
\iflncs{}\else \newtheorem{proposition}{Proposition}\fi
\newtheorem{proposition_a}{Proposition}[section]
\iflncs{}\else \newtheorem{theorem}{Theorem}\fi
\newtheorem{theorem_a}{Theorem}[section]
\iflncs{}\else \newtheorem{corollary}{Corollary} \fi
\newtheorem{corollary_a}{Corollary}[section]
\iflncs{}\else \newtheorem{lemma}{Lemma}\fi
\newtheorem{lemma_a}{Lemma}[section]
\iflncs{}\else \newtheorem{remark}{Remark}\fi
\iflncs{}\else \newtheorem{example}{Example}\fi

\ifmai

\iflncs{}\else \newtheorem{theorem}{Theorem}\fi

\iflncs{}\else \newtheorem{corollary}{Corollary} \fi

\iflncs{}\else \newtheorem{lemma}{Lemma}\fi

\fi

\newtheorem{problem}{Problem}

\newcommand{\nat}{\mathbb{N}}

\newcommand{\defi}{\stackrel{\triangle}{=}}

\newcommand{\range}{\mathrm{range}}
\newcommand{\dom}{\mathrm{dom}}

\newcommand{\es}{\emptyset}

\newcommand{\vsp}{\vspace*{.2cm}}

\newcommand{\reals}{{\mathbb{R}}}

\newcommand{\ode}{{\sc ode}}

\newcommand{\vecx}{\mathrm{\mathbf{x}}}
\newcommand{\xx}{\mathrm{\mathbf{x}}}

\newcommand{\zz}{\mathrm{\mathbf{z}}}

\newcommand{\vv}{v}
\newcommand{\parv}{\lambda}
\newcommand{\aalpha}{ {\bm{\alpha}}}
\newcommand{\bbeta}{ {\bm{\beta}}}
\newcommand{\ggamma}{ {\bm{\gamma}}}
\newcommand{\field}{\mathbb{R}}
\newcommand{\lie}{\mathcal{L}}
\newcommand{\ide}[1]{\big\langle\,\, #1 \,\,\big\rangle}

\renewcommand{\aa}{\mathrm{\mathbf{a}}}

\newcommand{\rk}{\mathrm{rk}}
\newcommand{\ivp}{\mathbf{\Phi}}

\newcommand{\var}{\mathrm{\mathbf{V}}}
\newcommand{\Id}{\mathrm{\mathbf{I}}}
\newcommand{\impl}{\longrightarrow}

\newcommand{\gdc}{\textsc{post}}
\newcommand{\pre}{\textsc{pre}}
\newcommand{\post}{\textsc{post}}

\usepackage{scalerel}
\stackMath
\newcommand\bighat[1]{%
\savestack{\tmpbox}{\stretchto{%
  \scaleto{%
    \scalerel*[\widthof{\ensuremath{#1}}]{\kern-.6pt\bigwedge\kern-.6pt}%
    {\rule[-\textheight/2]{1ex}{\textheight}}
  }{\textheight}%
}{0.5ex}}%
\stackon[1pt]{#1}{\tmpbox}%
}

\begin{document}

\iflncs
\title{Complete algorithms for algebraic strongest postconditions\\ and
weakest preconditions in  polynomial \ode s\thanks{Author's address: Michele Boreale, Universit\`a di Firenze,
Dipartimento di
Statistica, Informatica, Applicazioni (DiSIA) ``G. Parenti'', Viale Morgagni 65, I-50134
 Firenze, Italy. E-mail: \protect\url{michele.boreale@unifi.it}.}
}
\author{Michele Boreale}
\institute{Universit\`a di Firenze
}
\else
\title{\textbf{\Large Complete algorithms for algebraic strongest postconditions\\ and
weakest preconditions in  polynomial \ode s}\thanks{
Author's address: Michele
Boreale, Universit\`a di Firenze, Dipartimento di Statistica,
Informatica, Applicazioni (DiSIA) ``G. Parenti'', Viale Morgagni 65,
I-50134
 Firenze, Italy. E-mail: \protect\url{michele.boreale@unifi.it}.}}
\author{Michele Boreale\\Universit\`a di Firenze\\Dipartimento di Statistica,
Informatica, Applicazioni (DiSIA) ``G. Parenti''
}

\fi

\date{}
 \maketitle
\iflncs\thispagestyle{plain}\fi

\begin{abstract}
A system of polynomial ordinary differential equations (\ode s) is
  specified via a vector  of multivariate
polynomials, or vector field,   $F$.
A safety assertion $\psi\impl[F]\,\phi$
means that the   trajectory  of the system will lie  in a subset $\phi$ (the postcondition)
 of the state-space, whenever the initial state belongs to a subset $\psi$
  (the precondition). We consider the case when $\phi$ and $\psi$ are
  \emph{algebraic varieties},
  that is,   zero sets of polynomials. In particular,
  polynomials specifying the postcondition
  can be seen as  a system's  conservation laws  implied by $\psi$.
  Checking the validity of algebraic safety assertions is a fundamental problem in,
  for instance, hybrid systems.
We consider a generalized version  of this problem,
   and offer an algorithm that,
   given   a user specified polynomial set $P$ and an algebraic
    precondition $\psi$, finds the
   largest subset of polynomials in $P$  implied by
       $\psi$ (relativized strongest  postcondition).
  Under certain assumptions on $\phi$,
    this
     algorithm  can also be used to find the largest algebraic invariant included in $\phi$
     and the  weakest algebraic precondition   for $\phi$. Applications to continuous semialgebraic
     systems are also considered.
    The effectiveness of the  proposed algorithm
is demonstrated on several case studies from the literature.
\end{abstract}



\noindent
\textbf{Keywords}: Ordinary differential equations, postconditions, preconditions, invariants, Gr\"{o}bner bases.

\section{Introduction}\label{sec:intro}
In recent years, there has been a renewed interest in
 computational models based on ordinary differential equations (\ode s), in such diverse fields as
System Biology \cite{RifBio2} and stochastic systems \cite{Fluid}.
In particular, starting from \cite{San04}, the field of hybrid systems  has witnessed the emergence of a novel class of
 formal methods based on concepts from Algebraic Geometry --
 see e.g. \cite{Tiwa04,San10,Pla14}
and references therein.

A system of \ode s can be seen as specifying the evolution over time, or \emph{trajectory},
of certain variables of interest $x_1,...,x_N$,  describing for instance  physical
  quantities.
A fundamental problem in many fields is being able to prove or to disprove
  assertions of the following type.
For each    initial state in a given set $\psi\subseteq \reals^N$ (the precondition),
 the resulting system's trajectory will lie in a given set $\phi\subseteq \reals^N$ (the postcondition).
This is a \emph{safety assertion} that, using
 a notation akin to Platzer's Dynamic Logic, we can
write  as $\psi\impl[F]\,\phi$, where $F$ is the vector field specifying the system.
Evidently,   safety assertions
can be considered as a continuous counterpart of
 Hoare's triples in   imperative programs --- see \cite{Pla12}.

Here we are primarily interested in the case where both  $\psi$ and  $\phi$
are algebraic varieties, that is they are specified as zeros of
(multivariate) polynomial sets, and the drifts $f_i$ in
$F=(f_1,...,f_N)$  are  polynomials themselves.
Although (sets of) trajectories can rarely be
represented \emph{exactly} as algebraic varieties, these provide
overapproximations that may be useful in practice. In a valid safety
assertion,  the   polynomials specifying the postcondition $\phi$
can be seen as system's \emph{conservation laws} (for instance
energy or mass conservation\footnote{More precisely, when the precondition is  $\psi=\reals^N$, conservation laws
in our sense coincide  with what are known in Physics as   \emph{first
integrals of motion}, up to an additive constant.}) that are
 implied by the precondition $\psi$. Driven by the analogy with Hoare's triples, we find it   natural
to generalize  the problem of checking
the assertion $\psi\impl[F]\,\phi$ in two distinct ways. (1) Strongest postcondition: given a precondition
 $\psi$, find the smallest $\phi$ such that
the assertion is valid; (2) weakest precondition: given a postcondition $\phi$, find the largest
$\psi$ such that the assertion is valid.   Problem (1)
   amounts to characterizing       $I_\psi$,
   the set of \emph{all} polynomials invariants (conservation laws)
  implied by   $\psi$. The difficulty of  (1)
 motivates the introduction of a  {relativized} version of this
  problem: for a user specified
polynomial set $P$,   compute $P\cap I_\psi$. We call this a relativized
strongest postcondition.
Depending on $P$, computing this can be a lot easier   than computing the whole $I_\psi$.

We  offer a complete algorithm, called $\post$, that computes
   relativized strongest postconditions. 
In particular, this problem will be
considered in the case where the set $P$
 is specified via a polynomial template.
This way, for example, one can find  at once all polynomial conservation laws
of the system up to a given degree. As a byproduct of the $\post$ algorithm,
we also get  the weakest
algebraic invariant that implies all laws in $P\cap I_\psi$.  The $\post$ algorithm
is  based on
building ascending chains of polynomial ideals: these represent, basically,
 more and more refined overapproximations of the (relativized)  strongest postcondition.
  The proof  of correctness and termination relies on  a few concepts  from
   Algebraic Geometry, notably
Gr\"{o}bner bases \cite{Cox}. 
We will demonstrate the effectiveness of $\post$
   reporting the outcomes of a few experiments  we have
conducted on nontrivial systems taken from the literature, based on a Sagemath/Python
  implementation. 
Wherever possible, we will compare our results
with those obtained by other authors.

Focusing  on algebraic
sets, as we do,  does not necessarily imply that one is
limited to algebraic safety properties: in fact, by considering a family of
varieties depending
on a set of parameters, one can often get a good approximation  of
    a \emph{semialgebraic} set of interest. In our case, the parameters
    will basically correspond to possible
 initial conditions of the system. In this way, we will show that safety verification
 of continuous systems with   semialgebraic initial and unsafe sets
 is possible. 
 Formally, this will require  embedding the
original system into a larger space, with the introduction of
auxiliary variables,  and  relating these new variables to the
original ones using a suitable precondition.

The present paper    builds on
 our previous work  \cite{Fossacs17}, 
which deals with simple initial values
problems, where the precondition $\psi$ always consists of a
singleton. This restriction prevents one from dealing with the most
interesting continuous systems, such as semialgebraic systems. In particular,
 the concept  of weakest precondition is absent from \cite{Fossacs17}.
In the concluding section, we will   discuss relations with this work,
as well as with recent contributions from other authors, dealing with
invariant generation for polynomial \ode s  in the context of continuous and hybrid
systems, notably \cite{Pla14} and \cite{Kong}.

\paragraph{Structure of the paper} The rest of the paper is organized as follows.
The necessary mathematical preliminaries,
including polynomial differential equations and a few
facts from Algebraic Geometry, are introduced in Section \ref{sec:prel}, while
Section \ref{sec:ASA} introduces algebraic safety assertions and invariants.
In Section \ref{sec:problem2}, the main technical results are presented: the
 $\post$ algorithm and the proof of (relative) completeness.
 A more algorithmic presentation of $\post$ and computational issues connected with
   real radicals are discussed in Section \ref{sec:compasp}.
A few experiments  on systems drawn from the literature   are
described   in Section
\ref{sec:experiments}. An application to semialgebraic systems, together
with further examples, is the subject of Section \ref{sec:semialg}.
We round off the technical development so far with a discussion in Section \ref{sec:discussion}.
Related works are reviewed  in   Section \ref{sec:concl}.
For the sake of readability, a few technical proofs and some
additional technical material have been confined to  four separate
Appendices (\ref{app:pre}, \ref{app:proofs}, \ref{app:exp} and \ref{app:Psatz}).


\section{Preliminaries}\label{sec:prel}
We review a few preliminary notions about \ode s, polynomials, Lie derivatives and Algebraic Geometry.
\paragraph{Polynomial \ode s}
Let us fix an integer $N\geq 1$ and a set of $N$ distinct variables
$x_1,...,x_N$. We  will denote by $\vecx$ the
vector $(x_1,...,x_N) $.
We let $\field[\vecx]$ denote the set of multivariate polynomials in
the  variables $x_1,...,x_N$ with coefficients in $\field$, and let
$p,q$ range over it. Here we regard polynomials as syntactic
objects. Given an integer $d\geq 0$, by $\field_d[\vecx]$ we denote
the set of polynomials of degree $\leq d$.  As an example,
$p=2xy^2+(1/5)wz+yz+1$ is a polynomial of degree $\deg(p)=3$, that
is $p\in \field_3[x,y,z,w]$, with monomials $xy^2$, $wz$, $yz$ and
1. Depending on the context, with a slight  abuse of notation it may
be convenient to let a polynomial denote   the induced function
$\field^N \rightarrow
 \field$, defined as expected: for $\vv\in \field^N$, $p(\vv)\in \reals$ denotes the value obtained by
 evaluating $p$ at $\vv$.
 In particular, $x_i$ can be seen as denoting the projection on the $i$-th
 coordinate.

A (polynomial) \emph{vector field} is a  
vector of  $N$ polynomials,  $F=(f_1,...,f_N)$, seen as a function
$F:\mathbb{R}^N\rightarrow \mathbb{R}^N$. Throughout the paper, all definitions and statements refer
to an arbitrarily fixed polynomial vector field $F$ over a $N$-vector $\xx$.
The vector field $F$ and an initial condition $\vv_0\in \field^N$
together  define   an \emph{initial value problem} $\ivp=(F,\vv_0)$,  often written in the
following form
\begin{equation}\label{eq:ivp}
\ivp:\;\left\{\begin{array}{rcl}
{\dot\vecx(t)} & = & F(\vecx(t))\\
\vecx(0) & = & \vv_0\,.
\end{array}
\right.
\end{equation}
The functions $f_i$  in $F$ are called \emph{drifts} in this
context. A \emph{solution} to this  problem is a  differentiable
function $\vecx(t): D\rightarrow \field^N$,  for some nonempty open
interval $D\subseteq \mathbb{R}$ containing 0, which fulfills the
above two equations, that is: $\frac d{dt}\vecx(t)=F(\vecx(t))$ for
each $t\in D$  and $\vecx(0)=\vv_0$. By the Picard-Lindel\"{o}f
theorem \cite{PL}, there exists a nonempty open interval $D$
containing 0, over which there is a \emph{unique} solution, say
$\vecx(t)=(x_1(t),...,x_N(t))$,  to the problem. In our case,  as
$F$ is infinitely   differentiable,   the solution is seen to be
\emph{analytic} in $D$: each $x_i(t)$  admits a Taylor series
expansion   in a neighborhood of 0. For definiteness, we will take
the domain of definition  $D$ of $\xx(t)$  to be the largest
 open interval  where the Taylor expansion from 0 of each
of the $x_i(t)$ converges  (possibly $D=\field$). The resulting
vector function of $t$, denoted $\vecx(t)$, is called the \emph{time
trajectory} of the system. Note that both the time trajectory and
its   domain of definition do depend in general on the
initial $\vv_0$. We shall write them  as $\vecx(t;\vv_0)$
and $D_{\vv_0}$, respectively, whenever we want to make this
dependence explicit in the notation.

For  any polynomial $p\in
\field[\vecx]$,  the function $p(\vecx(t))  :D\rightarrow \field$,
obtained by composing $p$ as a
function with the time trajectory $\vecx(t)$, is    {analytic}: 
 we   let $p(t)$ denote
  the extension of this function 
 over the largest open interval of convergence
 (possibly coinciding with $\field$) of its Taylor expansion from 0.
We will call $p(t)$     the \emph{polynomial behaviour induced by
$p$ and by} the initial value problem \eqref{eq:ivp}. Again, fixing $N, \xx$ and $F$ once and for all,
we shall write $p(t;\vv_0)$ when we want to emphasize the dependence of this function
on the initial value $\vv_0$.

\paragraph{Lie derivatives}
Given a differentiable function $g:E\rightarrow \mathbb{R}$, for
  some  open set $E\subseteq \mathbb{R}^N$, the \emph{Lie
derivative of $g$ along $F$} is  the function $E\rightarrow
\mathbb{R}$ defined as:
$\lie_F(g)   \defi   \langle \nabla g, F\rangle = \sum_{i=1}^N
(\frac{\partial g}{\partial x_i} \cdot f_i)$.
The Lie derivative of   the sum $h+g$ and product  $h\cdot g$
functions obey  the familiar rules
\begin{eqnarray}
\lie_F(h+ g) & = & \lie_F(h)+\lie_F(g) \label{eq:liesum}\\
\lie_F(h\cdot g) & = & h\cdot \lie_F(g)+\lie_F(h)\cdot
g\,.\label{eq:lieprod}
\end{eqnarray}
Note that $\lie_F(x_i)=f_i$. Moreover if $p\in \mathbb{R}_d[\vecx]$
then $\lie_F(p)\in
\mathbb{R}_{d+d'}[\vecx]$, 
for some integer $d'\geq 0$ that depends on $d$ and on $F$. This
allows us to view the Lie derivative  of   polynomials along a
polynomial field $F$  as a purely syntactic mechanism, that is as a
function $\lie_F: \field[\xx]\rightarrow \field[\xx]$ that does not
assume anything about the solution of \eqref{eq:ivp}. Informally, we
can view $p$ as a program, and taking the  Lie derivative of $p$ can be
interpreted as  unfolding the definition  of each variable  $x_i$,
according to the equations in \eqref{eq:ivp} and to the formal rules
for product and sum differentiation, \eqref{eq:liesum} and
\eqref{eq:lieprod}. More generally, we can define inductively $\lie^{(0)}_F(p)\defi p$ and
$\lie^{(j+1)}_F(p)\defi \lie_F(\lie^{j}_F(p))$. 

\begin{example}\label{ex:running0}{\em
The following system,
borrowed from \cite{Kong}, will be used as a running example.
Consider $N=2$, $\xx=(x,y)$ and the   vector field $F=(y^2,xy)$.
  Let $p=x-y$. Examples of   Lie derivatives are
$\lie_F(p)=y^2-xy$ and $\lie^{(2)}_F(p)=2xy^2 -x^2y  - y^3$.
}\end{example}


In what follows, for any $\vv_0\in\reals^N$ we let $p(\vv_0)$ denote the real number obtained
by evaluating $p$ at $\vv_0$;   recall that $p(t;\vv_0)$ denotes
the function   $p(\xx(t;\vv_0))$, defined for $t$ in a suitable
neighborhood of the origin. We shall often abbreviate the
syntactic Lie derivative $\lie^{(j)}_F(p)$ as $p^{(j)}$,  and shall
omit the subscript ${}_F$ from $\lie_F$ when clear from the context.
The connection between Lie derivatives of $p$ along $F$ and the
initial value problem \eqref{eq:ivp} is given by the equations
below,
 which can be readily checked.
\begin{eqnarray}
p(t;\vv_0)_{|t=0} &  =  & p(\vv_0) \label{eq:initial}\\  
\frac{d\,}{dt}p(t;\vv_0) & = &
\left(p^{(1)}\right)(t;\vv_0)\nonumber\,.
\end{eqnarray}
More generally, we have the
following equation  for the $j$-th derivative of $p(t;\vv_0)$
($j=0,1,...$):
\begin{eqnarray}
\label{eq:liederj} \frac{d^j\,}{dt^j}p(t;\vv_0) & = & \left(p^{(j)}
\right)(t;\vv_0)\,.
\end{eqnarray}

\paragraph{Algebraic Geometry preliminaries}\label{sub:algeom}
We quickly review a few  notions from  Algebraic Geometry that will
be used throughout the paper. A comprehensive treatment of these
concepts can be found for instance in Cox et al.'s excellent
textbook \cite{Cox}. A set of polynomials $I\subseteq \reals[\xx]$
is an \emph{ideal} if: (1) $0\in I$ and (2) $p_1,...,p_m\in I$ and
$h_1,...,h_m\in \reals[\xx]$ implies $\sum_{i=1}^m h_ip_i\in I$. The
ideal generated by a set $P\subseteq \reals[\xx]$ is defined as
\begin{eqnarray*}
\ide{P}&\defi &
\left\{\sum_{i=1}^m h_ip_i\,:\,m\geq 0\text{ and }
h_i\in\reals[\xx], \,p_i\in P \text{ for }i=1,...,m\right\}\,.
\end{eqnarray*}
This
is the smallest ideal containing $P$ and as a consequence $\ide{\ide
P}=\ide{P}$. Given an ideal $I$, a set $P$ such that $I=\ide P$ is
said to be a set of \emph{generators} for $I$. Hilbert's basis theorem implies
that: (a) any ideal $I\subseteq \reals[\xx]$ has a finite set of generators; (b)
any infinite ascending chain of ideals $I_0\subseteq
I_1\subseteq\cdots$ stabilizes in a finite number of steps
(\emph{ascending chain condition}).   Once a total \emph{monomial order}
(e.g. lexicographic; see also Appendix \ref{app:proofs}) is fixed, a   multivariate version of
polynomial division naturally arises --- see \cite{Cox} for the
precise definition. A \emph{Gr\"{o}bner basis} of an ideal $I$ (w.r.t. a
fixed monomial order) is a finite set of generators $G$ of $I$ such that for any
polynomial $p\in \reals[\xx]$ the \emph{remainder} of the division
of $p$ by $G$, $r=p\bmod G$, enjoys following property: 
$p\in
I$ iff $r=0$. An alternative definition is that the
leading monomial (greatest in the monomial order)  of each $p\in I$ is divisible
by the leading monomial of some $g\in G$.
Given a Gr\"{o}bner basis $G$ of
$I$, the \emph{ideal membership} problem $p\in I$ can be
decided\footnote{Provided the involved coefficients can be
finitely represented, for instance are rational.} by just checking if $p\bmod G=0$.
Ideal inclusion $I\subseteq J$ can be decided similarly.  There are
algorithms that, given a finite $P$ and  a
monomial order, compute a Gr\"{o}bner basis $G$ such that $\ide G
=\ide P$: e.g. Buchberger's \cite{Buch} and Faugeres' F4 and F5 algorithms \cite{Fau4,Fau5}.
The worst-case time complexity of these algorithms is   exponential in the number
of variables, so this  computation is   potentially expensive.

The geometric counterpart of polynomial sets are algebraic varieties.
Given a set of polynomials $P\subseteq \reals[\xx]$, the set of
points in $\reals^N$ which are roots of all polynomials in $P$
\begin{eqnarray*}
\var(P)& \defi & \{\vv\in \reals^N: p(\vv)=0\text{ for each }p\in
P\}
\end{eqnarray*}
is the \emph{algebraic}\footnote{Some authors use
\emph{affine}.} \emph{variety} represented by $P$.
Ideals and algebraic varieties are connected as follows. For any set $A\subseteq \reals^N$, the
 set of polynomials that vanish
on $A$
\begin{eqnarray*}
\Id(A) & \defi & \{p\in \reals[\xx]\,:\,p(\vv)=0\text{ for each }\vv\in A\}
\end{eqnarray*}
is the ideal induced by $A$. Note that both $\var$ and $\Id$ are
inclusion reversing: $P\subseteq Q$ implies $\var(P)\supseteq
\var(Q)$, and $A\subseteq B$ implies $\Id(A)\supseteq \Id(B)$. For
$A$ an algebraic variety and $J$ an ideal, it is easy to see that
 $\var(\Id(A))=A$ and that $\Id(\var(J))\supseteq J$; if the equality  $\Id(\var(J))= J$
 holds, $J$ is said to be a \emph{real radical}.
  We will have in general
     more than one ideal $J$ representing $A$, that is such that $\var(J)=A$. 

\section{Algebraic safety assertions and invariants}\label{sec:ASA}
We will be interested in \emph{safety assertions} of the following
type, where $\psi,\phi\subseteq \reals^N$ are user specified
algebraic varieties, which we call the  \emph{pre} and
\emph{postcondition}, respectively. Each of them is specified by a
set of polynomials, that is, we will have $\phi=\var(P_1)$ and $\psi=\var(P_2)$ for
some $P_1,P_1\subseteq \reals[\xx]$.
\begin{equation}\label{eq:safety}
\text{Whenever $\vv_0\in \psi$ then for each $t\in D_{\vv_0}$,
$\xx(t;\vv_0)\in \phi$.}
\end{equation}\noindent
The above assertion means that every trajectory starting in the
precondition $\psi$ will stay in the postcondition $\phi$; hence
necessarily $\psi\subseteq \phi$ for the assertion to hold.
Using a notation akin to Platzer's Dynamic Logic's   \cite{Pla12}, the  safety
assertion \eqref{eq:safety} will be abbreviated as
\begin{eqnarray}\label{eq:safetypla}
\psi & \impl  & [F]\;\phi\,.
\end{eqnarray}
A common technique for proving  \eqref{eq:safetypla} is finding an
algebraic variety $\chi$ such that $\psi\subseteq \chi\subseteq
\phi$ and $\chi$ is an \emph{algebraic invariant} for the vector
field $F$, that is it satisfies $\chi   \impl [ F]\;\chi$. The
invariance condition means that all trajectories starting in $\chi$
must remain in $\chi$. 

Let us now introduce two
 distinct generalizations of the  problem of checking the
safety assertion
\eqref{eq:safetypla}. These are the problems we will
 actually try to solve. 
 In what follows, ``finding'' an algebraic variety means building a  finite set of
polynomials representing it.  Also note that, for varieties, ``smallest'' means ``strongest'', and
``largest''  means ``weakest''.

\begin{problem}[strongest postcondition]\label{pro:user-pre}
 Given an algebraic variety $\psi$,
    find $\phi_\psi$, the \emph{smallest} algebraic variety $\phi_\psi$ such that
  \eqref{eq:safetypla}  is true when $\phi=\phi_\psi$.
\end{problem}

Note that $\phi_\psi$ always exists and is the
intersection of all the varieties $\phi$ such that $\psi\impl
[F]\,\phi$. Finding $\phi_\psi$   amounts to
 building   (a basis of) an appropriate ideal $I$ such that $\var(I)=\phi_\psi$.
One such ideal is
\begin{eqnarray}\label{eq:Ipsi}
I_\psi& \defi & \Id(\phi_\psi)\, .
\end{eqnarray}
%
Currently,  we do not know how to compute $I_\psi$, or any other
polynomial representation of $\psi$. 
This
 motivates the introduction of a relaxed, or relativized, version of the previous
problem. In this version,   a user specified set of polynomials $P$
is used to tune the strength, hence precision, of the postcondition.

\begin{problem}[strongest postcondition, relativized]\label{pro:user-pre-rel}
Given  a polynomial set  $P\subseteq \reals[\xx]$ and an algebraic variety $\psi$,
    find a finite representation of $P \cap I_\psi $.
\end{problem}

Of course, we have that $\var(P\cap I_\psi)\supseteq
\var(I_\psi)=\phi_\psi$, which implies that $\psi\impl
[F]\,\var(P\cap I_\psi)$. In other words, $  P \cap I_\psi$ represents
an overapproximation of the strongest postcondition. There is
another meaningful way of generalizing the problem of checking
\eqref{eq:safetypla}.

\begin{problem}[weakest precondition]\label{pro:user-post}
 Given an algebraic variety $\phi$,
    find $\psi_\phi$, the \emph{largest} algebraic variety   such that
  \eqref{eq:safetypla}  is true when $\psi=\psi_\phi$.
\end{problem}

Let us now comment briefly on the relations existing between the
above introduced problems. It is not difficult to see that  being able to solve either of Problem
\ref{pro:user-pre} or Problem \ref{pro:user-post} implies one is able
to  check  \eqref{eq:safetypla} for
\emph{given} $\psi$ \emph{and} $\phi$, based on the fact that one
knows how to check inclusion between two varieties
  (see   Section \ref{sec:prel}). Indeed,  wanting to check the assertion
$ \psi \impl [F]\;\phi $, one may either check that $\phi\supseteq \phi_\psi$ or
 that $\psi\subseteq \psi_\phi$.
  The relativized Problem \ref{pro:user-pre-rel}
  too is more general than checking \eqref{eq:safetypla}. Indeed,
wanting to check  $ \psi \impl [F]\;\phi $, one may let
 $P=Q$ in Problem \ref{pro:user-pre-rel} and then
check if $P$ is included in the computed  $P\cap I_\psi$,  that is
if $P\subseteq I_\psi$.


\begin{example}\label{ex:running1}{\em
Let us reconsider the vector field $F$ of Example \ref{ex:running0}.
The variety $\psi=\var(\{p\})=\var(\{x-y\})$ is the line $x=y$.
Consider $\phi=\var(\{q\})$ where $q=x^2-xy$. Let $P$ be the set of
all polynomials of degree $\leq 2$. We can consider the following
problems. (a) Decide whether   $\psi\impl[F]\,\phi$; (b) find a
finite representation of $P\cap I_\psi$, that is all the
conservation laws of degree at $\leq 2$ that are satisfied, for each
initial state in the line $x=y$ (relativized strongest
postcondition); (c) find a finite representation of the largest
algebraic variety $\psi_\phi$ such that $\psi_\phi\impl[F]\,\phi$
(weakest precondition). Note that
    solving (b) also yields a solution of (a). 
}\end{example}

Concerning the weakest precondition Problem \ref{pro:user-post}, we
note that
 a   simple algorithm consists in   collecting all algebraic conditions
 ensuring that   the derivatives at any order of the polynomials
 specifying $\phi$ vanish.
Specifically, assuming $\phi=\var(P)$ for a user defined, finite set
$P\subseteq \reals[\xx]$,   one   considers the chain of sets
$P_0\defi P$, $P_{j+1}=P_j\cup \lie_F(P_j)=P_j\cup\{\lie_F(p):p\in
P_j\}$, until the least $m$ such that $\ide{P_{m+1}}=\ide{P_m}$,
where the last equality can be checked via Gr\"{o}bner bases
computation. Then one has $\psi_\phi=\var(P_m)$. Termination and
correctness of this algorithm can be easily derived from the results
presented, for example, in \cite{Liu,Pla14} (see also \cite{Novi}).
For the sake of completeness, we report a proof in Appendix \ref{app:pre}.
 In our experience, though, this simple algorithm tends to
scale badly as the number of variables grows, so   alternatives
are worthwhile to consider.

In the following
sections, we shall focus on  Problem 
\ref{pro:user-pre-rel}. 
In particular, we shall give a method, called $\post$, that works quite well in the
case when the polynomial set $P$ is specified by a polynomial
template.
Moreover, as a byproduct of this method, we will also get
 the weakest algebraic precondition for (and largest algebraic
 invariant included in)    $\var(P\cap I_\psi)$,
 which can be used to address Problem \ref{pro:user-post} as well. 
Finally, $\post$   will also give us a handle on the more
general and difficult  Problem \ref{pro:user-pre}.


\section{The $\post$ algorithm}\label{sec:problem2}
Recall from \eqref{eq:Ipsi} that $I_\psi$ is the ideal that induces the strongest algebraic postcondition
of the system at hand, for  a  user specified variety (precondition) $\psi$. Our goal is  to give
a method to effectively compute
 $P\cap I_\psi$, for a    polynomial  set $P$ which is itself user-specified.
Following a well-established tradition in the field of continuous and hybrid systems,
 we shall consider the case
when the user specifies  $P$   via  a polynomial \emph{template}, which we review   in the next paragraph. 

\lncstrue
\paragraph{Polynomial templates}\label{sub:templ}
\newcommand{\Lin}{{{\mathbb{L}}}}
Fix a tuple of $n\geq 1$ of distinct
\emph{template parameters}, say $\aa=(a_1,...,a_n)$, disjoint from $\xx$. Let
$\Lin(\aa)$, ranged over by $\ell$, be the set of \emph{linear
expressions} with coefficients in $\field$ and variables in $\aa$;
e.g. $\ell=5a_1 +42a_2-3a_3$   is one such expression\footnote{Note
that linear expressions with a constant term, such as $2+5a_1
+42a_2-3a_3$ are not allowed.
}. A \emph{template} \cite{San04} is a polynomial $\pi$ in
$\Lin(\aa)[\xx]$, that is, a polynomial with linear expressions as
coefficients. For example, the
following is a template: \iflncs $ \pi  =   (5a_1+(3/4) a_3)xy^2+
(7a_1+(1/5) a_2)xz + (a_2+42a_3)$. \else {
\begin{eqnarray*}
\pi & = & (5a_1+(3/4) a_3)xy^2+ (7a_1+(1/5) a_2)xz + (a_2+42a_3)\,.
\end{eqnarray*}}
\fi
Note that $\Lin(\aa)[\xx]\subseteq \reals[\aa,\xx]$, so, whenever convenient, we can consider
a template as a polynomial in this   larger ring. 
A \emph{template parameters valuation}  is a  vector
\begin{eqnarray*}
\parv & = & (\mu_1,...,\mu_n) \in \field^n\,.
\end{eqnarray*}
Given such a  $\parv$,
 we will let
$\ell[\parv]\in \field$ denote the result of replacing each template parameter
  $a_i$ with $\mu_i$, and evaluating the resulting expression;
we will let $\pi[\parv]\in \field[\xx]$ denote the polynomial obtained
by replacing each $\ell$ with $\ell[\parv]$ in $\pi$.  Given a set
$S\subseteq \field^n$, we let $\pi[S]$ denote the set
$\{\pi[\parv]\,:\, \parv\in S\}\subseteq \field[\xx]$. 
%
The (formal) Lie derivative of $\pi$ is defined as expected, once
linear expressions are treated as constants; note that $\lie(\pi)$
is still a template.
 It is easy to see that the following property is true, as a consequence of the fact
 that $a_1,..., a_n$ are treated as symbolic constants during differentiation: for
   each $\pi$ and $\parv$, one has $\lie(\pi[\parv])   =
   \lie(\pi)[\parv]$.
This property extends as expected to the $j$-th Lie derivative ($j\geq 0$):
\begin{eqnarray}
\vsp\lie^{(j)}(\pi[\parv]) & = & \lie^{(j)}(\pi)[\parv]\,.\label{eq:templder}
\end{eqnarray}
\lncsfalse

\paragraph{The   algorithm}\label{sub:refinement}
Given  a user specified algebraic variety $\psi$ (the precondition)
and a polynomial template $\pi$,  describing $P= \pi[\reals^n]$, our
objective is to compute the relativized strongest postcondition
$P\cap I_\psi$; recall that $I_\psi$ represents the strongest
algebraic postcondition. The following one is an important concept.

\begin{definition}[polynomial invariant]\label{def:polyinv}
Let us call   $p\in \reals[\xx]$ a \emph{polynomial invariant for
$F$ and  } $\vv_0$ if the function $p(t;\vv_0)$  is identically $0$.
\end{definition}

A polynomial invariant  expresses a law which is satisfied by the
solution of the initial value problem $(F,\vv_0)$, that is   a
conservation law.  We will rely on the following two lemmas. The first one
 is just a reformulation
of the definition of $I_\psi=\Id(\phi_\psi)$.
The easy proof of the second lemma is reported in  Appendix \ref{app:proofs}.

\begin{lemma}\label{lemma:Ipsi}
$I_\psi   = \{p\,:\,p \text{ is a polynomial invariant    for $F$ and
each }\vv_0\in\psi\}$.
\end{lemma}

\begin{lemma}\label{lemma:lie} Let $p\in \reals[\xx]$. Then $p$ is a
polynomial invariant for $F$ and $\vv_0$ if and only if
for each $j\geq 0$, $p^{(j)}(\vv_0)=0$.
\end{lemma}

The above  two lemmas  suggest  the following strategy to compute the
set $\pi[\reals^n]\cap I_\psi$. We should identify  those template parameter
valuations $\parv\in \reals^n$, such that     $\pi[\parv] $ is a
polynomial invariant for each $\vv_0\in \psi$ (Lemma \ref{lemma:Ipsi}). That is, those $\parv$'s such that
for each
$j\geq 0$ and for each $\vv_0\in\psi$, $\pi^{(j)}[\parv](\vv_0)=0$ (Lemma \ref{lemma:lie}). Or,
equivalently, $\pi^{(j)}[\parv]\in \Id(\psi)$ for each $j\geq 0$. For
each   $j\geq 0$, the last  condition imposes certain constraints on $\parv$, that
is on the   parameters  of the template $\pi^{(j)}$. In order to make
these constraints explicit, we shall rely on the following key
lemma. 
In the sequel we shall assume,
 over the polynomial ring $\reals[\aa,\xx]$,    a
 lexicographic
monomial order\footnote{This guarantees that, for any finite set $G\subseteq \reals[\xx]$,
$G$ is a Gr\"{o}bner basis in $\reals[\aa,\xx]$  if and only if it is in $\reals[\xx]$; see \cite[Ch.3,§1,Th.2]{Cox}. Any \emph{elimination} ordering  for the template parameters $a_i$ could as well be considered.} such that $a_i>x_j$ for each $i,j$.
 A proof of the lemma is reported in  Appendix \ref{app:proofs}.

\begin{lemma}\label{lemma:goebner} Let $G\subseteq \reals[\xx]$ be a Gr\"{o}bner basis.
Let $\pi$ be a polynomial template and $r= \pi\bmod G$. Then $r$  is
linear in the template parameters $a_1,...,a_n$. Moreover, for each $\parv\in \reals^n$, $\pi[\parv]\bmod  G
= r[\parv]$.
\end{lemma}

Fix a Gr\"{o}bner basis $G$ of $\Id(\psi)$. By the above lemma,
for a fixed $j$,   $\pi^{(j)}[\parv]\in \Id(\psi)$ exactly when $r_j[\parv]=0$, where $r_j=\pi^{(j)}\bmod G$. By
seeing $r_j$ as a polynomial in $\Lin(\aa)[\xx]$, the condition on $\parv$
\begin{eqnarray}\label{eq:condrj}
 r_j[\parv]&  =& 0
\end{eqnarray}
can be represented as a set of \emph{linear}  constraints on the template parameters
$\aa$: indeed, a polynomial is zero exactly when all of its
coefficients --- in the present case,  linear expressions in $\aa$ ---
are zero\footnote{For instance, if $\pi=(a_1+a_2)x_1+a_3x_2$ then $\pi[\parv]=0$ corresponds to the constraints
 $a_1=-a_2$ and $a_3=0$}. This discussion leads to the   method described below. 
We
first give a purely mathematical description of the method,
deferring the discussion of its computational aspects to Section
\ref{sec:compasp}.

The method can be seen as a generalization of  the double chain
algorithm of \cite{Fossacs17} to algebraic safety assertions; see
the concluding section for a discussion on the differences between the
two algorithms. The basic idea here is gradually refining the space
of template parameter valuations,
seen as a subset  of   $\reals^n$.
  More precisely,
the algorithm   builds two chains of sets: a descending chain of
vector spaces $V_i$, representing   spaces of   template parameter
valuations for which all the derivatives of $\pi$ up to order $i$ vanish
on   the points in $\psi$; and an (eventually) ascending chain of
ideals $J_i$, induced by the polynomials obtained from those parameter valuations. This   ideal chain is used
in the algorithm to detect the stabilization of the sequence, as discussed below. In
order to state the correctness of the result in the most general form, let us fix
an arbitrary ideal $I_0\subseteq \Id(\psi)$ and a
 Gr\"{o}bner basis
$G$ of $I_0$. Note that by admitting a $I_0$ smaller than
$\Id(\psi)$ we   actually allow for a weakening of the precondition.
For each $j\geq 0$, let $r_j\defi \pi^{(j)}\bmod G$. For each $i\geq
0$, consider the sets \vsp
\begin{eqnarray}
V_i & \defi & \{\parv\in \field^n\,:\,r_j[\parv] \text{ is the 0 polynomial, for }j=0,...,i\,\}\label{eq:Vi}\\[-3pt]
J_i & \defi & \left\langle\,\bigcup_{j=0}^i \pi^{(j)}[V_i]\,\right\rangle\label{eq:Ji}\,.
\vsp
\end{eqnarray}\noindent
It is easy to check that each $V_i\subseteq \field^n$ is a
vector space over $\field$ of dimension $\leq n$: this stems from
the linearity  in $\aa$ of the $r_j$ terms. Now let $m\geq 0$ be  the
least integer such that the following conditions are \emph{both}
true:
\begin{eqnarray}
V_{m+1} & = &V_m \label{eq:Vm}\\
J_{m+1} & = & J_m\,.\label{eq:Jm}
\end{eqnarray}\noindent
The algorithm returns $(V_m,J_m)$, written
$\gdc_F(\psi,\pi)=(V_m,J_m)$; we shall omit the subscript ${}_F$ when the vector field $F$
is clear from the context. Note that the integer $m$ is well
defined: indeed, $V_0\supseteq V_1\supseteq \cdots$ forms  an
infinite descending chain of finite-dimensional vector spaces, which
must stabilize in finitely many steps. In other words,  we can
consider the least $m'$ such that $V_{m'}=V_{m'+k}$ for each $k\geq
1$.
Then $J_{m'}\subseteq J_{m'+1}\subseteq\cdots$ forms    an infinite
ascending chain of ideals, which must
  stabilize at some $m\geq m'$. 
Therefore there must be some index $m$ such that \eqref{eq:Vm}  and
\eqref{eq:Jm} are both satisfied, and we   choose the least such
$m$.

\paragraph{Results} We start with an important concept, which is needed to state and prove the correctness
and completeness of $\post$.

\begin{definition}[invariant ideal]\label{def:invid}
A set of polynomials $J\subseteq \reals[\xx]$ is an \emph{ invariant
ideal} for the vector field $F$ if it is an ideal and
$\lie_F(J)\defi \{\lie_F(p):p\in J\}\subseteq J$.
\end{definition}

The next theorem states the correctness and relative  completeness
of \gdc. Informally, the algorithm   outputs a space $V$
such that $\pi[V]\subseteq I_\psi$, which is the largest such space if $I_0=\Id(\psi)$, and
 the smallest invariant ideal $J$ including $\pi[V]$. This invariant ideal also conveys important
 information about the system, as discussed later on in the section.
In order to prove the main theorem, we need a technical lemma, whose proof is reported in the Appendix \ref{app:proofs}.

\begin{lemma}\label{lemma:stab}
Let $V_m,J_m$ be the sets returned by the \gdc\ algorithm. Then 
for each $j\geq 1$, one has
 $V_m=V_{m+j}$ and $J_m=J_{m+j}$.
\end{lemma}


\begin{theorem}[correctness and relative completeness of $\gdc_F$]\label{th:corr}
Let $\psi$ be an algebraic variety, let $I_0\subseteq I_\psi$ be an ideal and $G$ be a  Gr\"{o}bner basis of $I_0$. For any polynomial template $\pi$, let
 $\gdc_F(\psi,\pi)=(V,J)$. Then
\begin{itemize}
\item[(a)] $\pi[V]\subseteq  \pi[\field^n]\cap I_\psi$. In
particular, $\pi[V]=  \pi[\field^n]\cap I_\psi$
 if $I_0=\Id(\psi)$;
\item[(b)] $J$ is  the smallest invariant  ideal  such that   $J\supseteq \pi[V]$. 
Moreover,  $J\subseteq I_\psi$. 
\end{itemize}
\end{theorem}
\begin{proof}
Let $(V,J)=(V_m,J_m)$ for some $m\geq 0$.
Concerning part (a), we first note that, by virtue of Lemma
\ref{lemma:Ipsi} and Lemma \ref{lemma:lie},
 $\pi[\parv]\in \pi[\field^n] \cap I_\psi
$ if and only if for each $j\geq 0$,   $(\pi[\parv])^{(j)}=
\pi^{(j)}[\parv]\in \Id(\psi)$ (here we have used property
\eqref{eq:templder}).
If $\parv \in V_m=V_{m+1}=V_{m+2}=\cdots$ (here
we are using Lemma \ref{lemma:stab}), then by definition, for each
$j\geq 0$, $r_j[\parv] = (\pi^{(j)})[\parv]\bmod G = (\pi[\parv])^{(j)}\bmod
G=0$ (here we have used again
 property \eqref{eq:templder}
and Lemma \ref{lemma:goebner}). That is, for each $j\geq 0$,
 $(\pi[\parv])^{(j)}\in I_0\subseteq \Id(\psi)$. This implies   (again by
Lemma \ref{lemma:Ipsi} and \ref{lemma:lie}) that $\pi[\parv]\in I_\psi$.
Assume now that $I_0=\Id(\psi)$ and let $\parv$ such that
 $\pi[\parv]\in \pi[\field^n] \cap I_\psi$, that is
 if for each $j\geq 0$,   $(\pi[\parv])^{(j)}=
\pi^{(j)}[\parv]\in \Id(\psi)$. That is, being $G$ a    Gr\"{o}bner basis of $\Id(\psi)$,
 $\pi^{(j)}[\parv]\bmod G= r_j[\parv]=0$  (the first equality here follows
from  Lemma \ref{lemma:goebner}), for each $j\geq 0$. This
assertion, by definition,
  means   that   $\parv \in V_j$ for each $j\geq 0$, hence in particular
$\parv \in V_m$.

Concerning part (b),  to prove that $J_m$ is the smallest invariant
ideal including $\pi[V_m]$, it is enough to prove the following: (1)
$J_m$ is an invariant ideal,   (2) $J_m\supseteq \pi[\field^n]\cap
I_\psi $,  and (3) for any invariant ideal $I$ such that $
\pi[\field^n] \cap I_\psi\subseteq I$, we have that $J_m\subseteq
I$. We first prove (1),  that  $J_m$ is an invariant ideal. Indeed,
for each $\parv \in V_m$ and each $j=0,...,m-1$, we have
$\lie(\pi^{(j)}[\parv])=\pi^{(j+1)}[\parv]\in J_m$ by definition,
while for $j=m$, since $\parv \in V_m=V_{m+1}$, we have
$\lie(\pi^{(m)}[\parv])=\pi^{(m+1)}[\parv]\in J_{m+1}=J_m$ (note
that in both cases we have used property \eqref{eq:templder}).
Concerning (2),  note that $J_m\supseteq \pi[V_m]=\pi[\field^n]\cap
I_\psi$, by
virtue of part (a). 
Concerning (3),  consider any invariant ideal $I\supseteq
\pi[\field^n]\cap I_\psi $.  We show by induction on $j=0,1,...$
that for each $\parv \in V_m$,  $\pi^{(j)}[\parv]\in I$; this will imply the
wanted statement.  Indeed, $\pi^{(0)}[V]=\pi[V]\in I_\psi \cap
\pi[\field^n]$, as $\pi[V_m]\subseteq I_\psi$ by (a). Assuming now
that $\pi^{(j)}[\parv]\in I$, by invariance of $I$ we have
$\pi^{(j+1)}[\parv]= \lie(\pi^{(j)}[\parv])\in I$ (again, we have used here
property \eqref{eq:templder}).

Finally, $J_m\subseteq I_\psi$ follows from the last statement and
from the fact that $I_\psi$, as clearly seen from Lemma
\ref{lemma:Ipsi} and \ref{lemma:lie},  is   an invariant ideal.
\end{proof}

\begin{example}\label{ex:running2}{\em
We reconsider the vector field $F$ of Example \ref{ex:running0}. Let us consider
$\psi=\var(\{x-y\})$. A Gr\"{o}bner basis of $\Id(\psi)$ is just $G=\{x-y\}$.
 We let $\pi$
be the complete template of degree 2 (described below). Running
$\post_F(\psi,\pi)$ means building the chain of sets
 $V_i,J_i$, for $i=0,1,...$.  Below,
$\parv=(\mu_1,...,\mu_6)\in \reals^6$ denotes a generic template parameters
valuation, while $e_j\in \reals^n$ denotes the $j$-th basis vector
in $\reals^n$, for $j=1,...,6$. With the help of
a the computer algebra system  (e.g. SageMath \cite{Sage}), we consider
 the successive Lie
derivatives of $\pi$ and their remainders $\bmod \,G$, as follows:
 {\small
\begin{itemize}
\item $\pi = a_6xy + a_5y^2 +a_4 x^2+ a_3 y+ a_2 x  + a_1$  and
  $r_0=\pi\bmod G = (a_4+a_5+a_6) y^2   + (a_2+a_3) y   + a_1$. Thus
$V_0=\{\lambda\in
\reals^6\,:\,\mu_4+\mu_5+\mu_6=0,\,\mu_2+\mu_3=0,\,\mu_1=0\}$. A
parametric representation of the elements of $V_0$ is
$(0,-\mu_3,\mu_3,-\mu_5-\mu_6,\mu_5,\mu_6)$, from which a basis  is
obtained  as $B_0=\{-e_2+e_3,-e_4+e_5,-e_4+e_6\}$. We let $J_0=\ide
{\pi[V_0]}=\ide{\pi[B_0] }$;
\item $\pi^{(1)} = a_6 x^2 y + 2( a_4 +   a_5) x y^2 + a_6 y^3 + a_3 x y + a_2 y^2$ and
  $r_1=\pi^{(1)}\bmod G = 2(a_4 + a_5 + a_6) y^3 + (a_2 + a_3) y^2$: the last equation
  can be checked directly by   equating $x$ and $y$, which is what $\pi^{(1)}\bmod \{x-y\}$
  means. We see that $r_1=0$ does not induce new constraints on the template parameters, that is
$V_1=V_0$. Next, we compute  a basis for
 $\pi^{(1)}[V_1]=\pi^{(1)}[V_0]$ as $ \pi^{(1)}[B_0]=\{ x y - y^2,0,
 x^2 y - 2 x y^2 + y^3 \}$; we can check (again, equating $x$ and $y$)
  that the last set is $\subseteq J_0$, which implies $
\pi^{(1)}[V_1]\subseteq J_0$. This in turn implies
$J_1=\ide{\pi[V_1]\cup\pi^{(1)}[V_1]}=J_0$: therefore also the ideal chain has stabilized.
\end{itemize}
}\noindent
Therefore both chains stabilize already at $m=0$ and
 $\gdc(\psi,\pi)=(V_0,J_0)$.  A Gr\"{o}bner basis of
$J_0$ is $G_0=G$.
}\end{example}

\begin{remark}[result template]\label{rem:convention}{\em
Given a template $\pi$ and
$\parv\in \field^n$, checking if $\pi[\parv]\in \pi[V]$ is equivalent to
checking if $\parv \in V$: this can be effectively done
 knowing  a basis $B$ of the vector space  $V$
(see Section \ref{sec:compasp}). In practice, it is computationally more convenient
to represent the whole set $\pi[V]$ returned by \gdc\  compactly in terms of a \emph{new}
 $n'$-parameters ($n'\leq n$) result template $\pi'$ such that $\pi'[\reals^{n'}]=\pi[V]$.
 For instance, in the previous example,   the
result template $\pi'=a_1 (y^2- x^2)  + a_2 (x y -  x^2) +     a_3 (y-   x)$
 represents $\pi[V_0]$, in the precise
 sense that  $\pi[V_0]=\pi'[\reals^3]$. The result template $\pi'$
 can in fact be built directly from $\pi$, by
propagating the linear constraints on $\aa$ \eqref{eq:condrj}  as
they are generated. This will be explicitly described when
discussing the algorithmic presentation in Section
\ref{sec:compasp}.
}\end{remark}

Note that, while
  typically the user will be interested in
 $\pi[V]$,   the ideal $J$ as well may contain useful information, such
 as higher order, nonlinear conservation laws.
The
  theorem below is about the meaning  of $J$ as an invariant and as a precondition.
The theorem relies on Theorem \ref{th:corr}  and on the following
lemma, 
stating that invariant
 ideals, on the polynomial side,
 precisely
 correspond to algebraic invariants.
 The proofs of both the lemma and   the theorem are reported in the     Appendix \ref{app:proofs}.

\begin{lemma}\label{lemma:invar} Consider a set $\chi\subseteq \reals^N$. Then $\chi$ is an
algebraic invariant for the
 vector field
$F$ if and only if there is an invariant ideal $J$ for $F$ such that
$\chi=\var(J)$.
\end{lemma}

\begin{theorem}[weakest algebraic invariant and precondition]\label{th:instructive}
For an algebraic variety $\psi$ and
a polynomial template $\pi$, let
 $\gdc_F(\psi,\pi)=(V,J)$ and $\phi=\var(\pi[V])$. 
 Then
\begin{itemize}
\item[(a)]
 $\var(J)$ is the largest algebraic invariant included in  $\phi$; and
\item[(b)] $\var(J)$ is the weakest precondition of $\phi$.
\end{itemize}
\end{theorem}

We stress that      Theorem \ref{th:instructive}(b)   provides a
means to solve Problem \ref{pro:user-post} (weakest precondition)  via
the \gdc\ algorithm. In fact, given $\phi$, it suffices to
 consider \emph{any}  precondition $\psi$ and template $\pi$ such that $\gdc(\psi,\pi)=(V,J)$ and
$\var(\pi[V])=\phi$: then $\var(J)$ is $\phi$'s weakest precondition.
In particular, $\psi$ may consists of a singleton.
Also note that the theorem does not require the equality $I_0=\Id(\psi)$.
An example of application of this technique is given below.
Other examples will be
  discussed in Section \ref{sec:experiments} (see in particular the Kepler laws example).

\begin{example}\label{ex:running3}{\em
We reconsider the vector field $F$ of Example \ref{ex:running0}. Let
$\phi=\var(\{q\})$ be given, where $q=x^2-xy$. We want to compute the weakest precondition
$\psi_\phi$ via $\post$. We choose  the trivial precondition
$\psi=\{(0,0)\}=\var(\{x,y\})$:   both $x(t; \,(0,0))$ and $y(t;\, (0,0))$ are identically 0,  making
$\psi\rightarrow[F]\,\phi$  a valid assertion. Now choosing   $\pi=a_1\cdot q$ and running $\post$,
we obtain $\post(\psi,\pi)=(V,J)$, where $V=\reals$ and necessarily $\var(\pi[V])=\phi$, and
$J=\ide{\{xy^2 - y^3\,,\, x^2 - xy\}}$. By Theorem \ref{th:instructive}(b),
$\psi_\phi=\var(J)=\var(\{x-y\})$.
}\end{example}

Finally, the following result of theoretical interest, shows
 that the whole ideal $I_\psi$  as well can be characterized in
terms of the $\gdc$ algorithm. For any $k\geq 0$,  the
\emph{complete polynomial template} of degree $k$ over a set of
variables $X$ is
$\pi  \defi \sum_{\alpha}a_\alpha \alpha$, where $\alpha$ ranges
over all monomials of degree $\leq k$ on the variables in $X$, and
$a_\alpha$ ranges over distinct template parameters.

\begin{corollary}[characterization of $I_\psi$]\label{cor:ipsi}
 Let $\psi$ be an algebraic variety. Let
 $k\geq 0$, $\pi_k$ be the complete template of  degree $k$ over  the variables in
$\xx$ and $(V,J)=\gdc(\psi,\pi_k)$.  For $k$ large enough, $J
=I_\psi$.
\end{corollary}
\vsp \noindent
\begin{proof}
By Hilbert's basis theorem, there is a finite set of polynomials $P$
such that $I_\psi=\ide P$. Therefore $I_\psi$ is the smallest ideal
containing $P$, and is also an invariant  ideal. Now let $k$ be the
maximum degree of polynomials in $P$, let $\pi_k$ be the complete
template of degree $k$ over all   variables, and $n$   the number of template
parameters in  $\pi_k$. As $P\subseteq \pi_k[\reals^n]$ and
$P\subseteq I_\psi$, we have
 $\pi_k[\reals^n]\cap I_\psi\supseteq P$.
Now let $(V,J)=\gdc(\psi,\pi_k)$. By Theorem \ref{th:corr}(b),
$J\supseteq \pi_k[V]=\pi_k[\reals^n]\cap I_\psi \supseteq P$, hence
$J\supseteq I_\psi$. On the other hand, again by Theorem
\ref{th:corr}(b), $J\subseteq I_\psi$.
Therefore $J=I_\psi$.
\end{proof}

\vsp We leave open the problem of computing a lower bound on the
degree $k$  that is needed to recover $I_\psi$. We end the section
with a remark on the expressive power of algebraic varieties.

\begin{remark}[expressive power]\label{rem:ghost}{\em
Algebraic varieties can in general provide  only overapproximations
of sets of initial states and trajectories.   However, the
expressive power of algebraic varieties can often be  significantly
enhanced by introducing auxiliary, or \emph{ghost} variables, in the
terminology of Platzer \cite{Pla12b}. These variables are used to
express properties of interest. We have found particularly
interesting the case when ghost variables are used to encode
\emph{generic} initial values of the system: apparently, keeping
track of such values allows for more expressive polynomial
invariants. This is illustrated by the example below. We will put
this technique into use in Section \ref{sec:experiments} and, in a
more systematic way, in Section \ref{sec:semialg}, where we shall deal with
semialgebraic systems. 
}\end{remark}

\begin{example}\label{ex:running5}{\em
Consider again the system of Example \ref{ex:running0}. With no
constraints on the initial states, that is with $\psi=\reals^2$, the
strongest postcondition is quite easily seen to be the trivial
$\phi=\reals^2$, that is $I_\psi=\{0\}$.
We build now a new system by  introducing two new variables
$x_0,y_0$,  together with the corresponding equations $\dot x_0=0$
and $\dot y_0=0$: this means  they represent (generic) constants ---
in effect, parameters. We consider  the precondition
$\psi=\var(\{x-x_0,y-y_0\})$, meaning that $x_0$ and $y_0$ represent
the (generic) initial values of $x$ and $y$, respectively. Using a
complete template $\pi$ of degree 2, we now get the nontrivial
result $\gdc(\psi,\pi)=(V,J)$ with $J=\ide{\{x_0^2 - y_0^2 - x^2 +
y^2\}}$ and $\pi[V]=\pi'[\reals]$, where $\pi'=a_1(x^2 - x_0^2 - y^2 + y_0^2)$. Here $J$ represents a valid
nontrivial invariant for every instantiation  of $x_0,y_0$. 
}\end{example}

\section{Computational aspects of \gdc}\label{sec:compasp}
We   consider here  some important computational aspects of the $\gdc$
algorithm. We will first derive  a more algorithmic presentation of the
abstract procedure introduced in Section \ref{sec:problem2}; then
discuss issues related to selecting  an appropriate ideal $I_0 \subseteq \Id(\psi)$ and a
Gr\"{o}bner $G$
basis for it.

\subsection{Algorithmic presentation}
When it comes to the effective implementation of $\post$, the first
aspect to  consider is how to finitely represent  the sets $V_i,J_i$.
Each subspace $V_i$  is spanned by a finite basis $B_i$, which can   in  principle
 be computed explicitly from  the linear constraints on the template parameters $a_1,...,a_n$ imposed by \eqref{eq:Vi}. From \eqref{eq:Ji} it is  then easy to check that
$\bigcup_{j=0}^i\pi^{(j)}[B_i]$ is a basis of $J_i$. The termination
conditions $V_i=V_{i+1}$ and $J_i=J_{i+1}$ can also be checked
effectively. In particular, checking   $J_i=J_{i+1}$ involves
 computing a Gr\"{o}bner basis  of $J_i$, a
potentially expensive operation, and checking if $\pi^{(i+1)}[B]\subseteq J_i$.
Fortunately, this need   not   be done at each step, but only if  actually
$V_i=V_{i+1}$, the latter a   relatively inexpensive check.

Rather than building the $V_i$'s  explicitly, computationally it is   more  convenient to represent them implicitly, via symbolic linear constraints on the template parameters $a_1,...,a_n$.
At each step $i\geq 0$, such constraints are generated from the condition     $r_i=0$  on the remainder (see \eqref{eq:condrj}), and are represented by a   substitution $\gamma$  that eliminates a few template parameters.  The constraints $\gamma$ are   propagated to all templates $\pi^{(j)}$ ($0\leq j\leq i$) generated so far.
  This discussion leads to Algorithm \ref{alg:post}. Note that program
  blocks are defined by indentation.
  We make use of a few auxiliary
  variables and functions, as detailed below.

\begin{algorithm}[t]
\caption{$\post$  }
\begin{algorithmic}[1]

\Statex \textbf{Input}: {\small  $F$ a vector field, $\pi$ a $n$-parameters template, $G\subseteq \Id(\psi)$ a   Gr\"{o}bner basis for $I_0$ }
\Statex \textbf{Output}: {\small $\pi'$ a $n'$-parameters template, $J$ an invariant ideal (s.t.   $\post_F(\psi,\pi)=(V,J)$ and $\pi'[\reals^{n'}]=\pi[V]$)}
\State $S:=[\,]$
\While{ \textbf{true} }
\State $r:=\pi\bmod G$
\State $\gamma:= \mathrm{solve}(r=0)$
        \If{$\gamma=\es$}\COMMENT{Check if $V_{i+1}=V_i$; equivalently if $r=0$}
            \State $GS:=\mathrm{instantiate}(S)$
            \State $J:= \ide{GS}$
            \If{$\mathrm{instantiate}(\{\pi\})\subseteq J$}\COMMENT{Check if  $J_{i+1}=J_i$; if yes, we have stabilization}
                \State \textbf{return} $(\mathrm{first}(S),J)$
            \EndIf
        \Else\COMMENT{$V_{i+1}\neq V_i$; vector spaces chain not stabilized}
            \State $\pi:=\gamma(\pi)$\COMMENT{Propagate constraints $\gamma$}
            \State $S:= \gamma(S)$
        \EndIf
\State $S:=\mathrm{append}(S,\pi)$
\State $\pi:=\lie_F(\pi)$
\EndWhile

\end{algorithmic}
\end{algorithm}\label{alg:post}

\begin{enumerate}
 \item $S$ is an initially empty list of polynomial templates, used to collect the successive Lie
 derivatives of $\pi$.
 The functions $\mathrm{first}( \cdot )$ and $\mathrm{append}( \cdot )$, defined on lists, have the usual
 interpretation: $\mathrm{first}(S)$ returns the first element of $S$, with the   proviso that
   $\mathrm{first}([\,])\defi 0$, the zero polynomial template, while $\mathrm{append}(S,\pi)$
  returns the list obtained by   appending
   $\pi$ to   $S$ as a last element.

\item   $\gamma$ is a  {substitution}, encoding linear constraints
existing among the template
parameters.
Formally, a substitution $\gamma$ is
  a finite partial
 map  from $\{a_1,...,a_n\}$ (template parameters) to $\Lin[\aa]$ (linear expressions),
  such that
 no parameter in $\dom(\gamma)$
occurs in any $\ell\in \range(\gamma)$. We write $\gamma(c)$ for the result of applying
$\sigma$ to every parameter occurring in   (the  expression, set,...) $c$.

 \item $\mathrm{solve}(r=0)$ returns a (minimal) substitution $\gamma$ such that
    $\gamma(r)=0$, the zero polynomial. 
    We insist  that   $|\dom(\gamma)|$ be minimal, that is eliminate  as few template parameters as possible.
    In linear algebraic terms,
    let  $\ell_1,..., \ell_K\in \Lin(\aa)$ be the distinct coefficients of $r$, where  $\ell_i=\sum_{j=1}^n
    c_{ij}a_j$. Let $C$ be the $K\times n$ real coefficients matrix of the $c_{ij}$'s.
    A template  parameter valuation $\lambda\in \reals^n$ makes $r[\lambda]$ the null polynomial if and only
    if $\lambda^T$ is a solution of the linear system  in the variables $\aa$,  $C\aa^T=0$.
    We insist that $\gamma$   describe the whole space $U\subseteq \reals^n$ of solutions of this system.
As $U$  has dimension  $n-\rk(C)$, this is equivalent to saying that
$\gamma(r)=0$ and $|\dom(\gamma)|=\rk(C)$.


\item  $\mathrm{instantiate}(L)\subseteq \reals[\xx]$, for $L$ a list or set of  templates,    returns a   finite generating set  of the vector space spanned by $\cup_{\pi'\in L}\,\pi'[\reals^{n}]$ in $\reals[\xx]$.
Specifically, letting   $e_j\in \reals^n$
  denote the $j$-th canonical basis vector ($1\leq j\leq n$), seen as a parameter valuation, we
  have (below 0 denotes the zero polynomial):
\begin{eqnarray*}
 \mathrm{instantiate}(L) & \defi &
 \{ \pi'[e_j]\,:\,\pi' \in L,\, \,j=1,...,n\}\setminus\{0\}\,.
\end{eqnarray*}
As an example, if $n=4$ and $\pi=  (2a_2-a_3)xy+a_3$ then
$ \mathrm{instantiate}(\{\pi\})=\{\pi[e_1],\pi[e_2],\pi[e_3],\pi[e_4]\}\setminus \{0\}=\{2xy,-xy+1\}$.
%
%
\end{enumerate}


The exact
theoretical complexity of this algorithm is  difficult to characterize, even
assuming, as we do here,    that the basis   $G$ s.t. $\ide{G}=I_0\subseteq \Id(\psi)$
has been precomputed.
But one can at least work out some very conservative bounds, as
follows.  Let us denote by $d$  the sum of the degree of $\pi$ and
of the maximal degree of polynomials in $F$,  and by $N$ the number
of variables.  We note that:  (a) each step  potentially involves
the computation of a Gr\"{o}bner basis, for which  known algorithms
have an exponential worst case time complexity upper bounded
approximately by $O(D^{2^N})$, where $D$  is the maximum degree in
the input polynomial set  (see \cite{Cox}); (b) the maximum degree
$D$ of the derivatives $\pi^{(j)}$'s occurring in $S$, for $0\leq j\leq m+1$,  is   bounded by
$(m+1)d$. Overall, this gives a worst case time complexity  of
approximately $O(m^{2^N+1}d^{2^N})$. Finally, according to a result
in \cite{Novi}, the number of steps $m$ before stabilization of an
ascending chain of ideals generated by successive Lie derivatives is
upper bounded by $d^{N^{O(N^2)}}$. 
One should stress that these are very conservative bounds.
%
%
A SageMath/Python \cite{Sage} implementation strictly adhering to Algorithm \ref{alg:post}
 is available\footnote{\url{https://github.com/micheleatunifi/postconditions/blob/master/Post.py}} that
 works reasonably well in a number of cases of practical interest; see Section \ref{sec:experiments}.
SageMath directly
provides an efficient implementation in exact rational arithmetic of the most important
auxiliary functions, such as linear constraints   and Gr\"{o}bner bases generation.

\subsection{The choice of $I_0$ and the real radical problem}
\newcommand{\lm}{\mathrm{LM}}
A crucial aspect   in the \gdc\ algorithm  is the choice of the
ideal $I_0\subseteq \Id(\psi)$ and the computation of a Gr\"{o}bner
basis $G$ for it. In the following discussion, we fix
 the following notation  and terminology:
 \begin{itemize}
\item $\psi=\var(Q)$,    the variety generated by   a finite (user specified) set
 $Q\subseteq \reals[\xx]$;
\item $I=\ide{Q}$, the ideal generated by $Q$;
\item $\Id(\psi)=\{p: p(v)=0 \text{ for each } v\in \psi\}\supseteq I$,
the \emph{real radical} of $I$.
\end{itemize}
 In the statement of
Theorem \ref{th:corr}(a),  equality, hence completeness,  is
guaranteed if    $I_0=\Id(\psi)$ is the real radical of $I$; otherwise only   soundness holds
in general.
%
Unfortunately, at present computing a set of generators $G$ for a  real radical
 appears to be, in the  {general} case,   computationally
infeasible. Below, we shall briefly discuss the state of the
art concerning this problem, then a special case where this computation is
 feasible, and what are the alternatives in cases where it is not.

A classical algorithm for computing real radicals is due to   Neuhaus
  \cite{Neuhaus}. This is also implemented
  as part of Singular's   \verb"realrad" library \cite[Sect.D.4.16]{Singular},
    accessible via SageMath \cite{Sage}.
The worst case asymptotic complexity of this algorithm is very high:
$D^{2^{O(N^2)}}$ (exact) arithmetic operations,
where $D$ is the maximum total degree
of the polynomials in the set $Q$.
Over the years there  have been
 improvements: let us just
  mention the algorithm by Lasserre et al.,
  based on semi-definite relaxations
   but  limited to ideals with zero dimensional
 varieties \cite{Lasserre}; and the   recent
 probabilistic method by El Din et al. \cite{Eldin}, which
 lowers  the asymptotic complexity to
 $|Q|^{O(1)}(ND)^{O(Nr2^r)}$, where $r$ is the dimension
 of the variety $\psi$. Despite these advances,
  the resulting algorithms
 appear  to be still totally impractical  but for very simple instances:
  \cite{Eldin} mentions
 an example with $N=9$ variables and maximum total degree $D=3$
 which is beyond Singular's capabilities
 and   requires 800s with their implementation.


Next, we consider  a simple special case of practical interest,
where it is trivial to build the real radical.  This case is
relevant to the auxiliary variables method mentioned in Remark
\ref{rem:ghost}, 
 and
 will be put into
systematic
 use when dealing with semialgebraic systems (Section \ref{sec:semialg}).
 The general idea is that
 $\psi$, as a precondition,
  equates each system  variable to a generic constant, or polynomial expression  thereof, which are the $g_i$s
  in the statement.
  This permits a quite uniform and general treatment of initial conditions.

\begin{proposition}\label{prop:special} Let $\xx=(x_1,...,x_N)$, let $1\leq m<N$ and
assume
 $Q=\{x_i-g_i\,:\,i=1,...,m\}$,
 where   $g_i\in \reals[x_{m+1},...,x_N]$. Then, for $\psi=\var(Q)$, we have   $\Id(\psi)=\ide{Q}$.
\end{proposition}
\begin{proof}
Fix  a lexicographic monomial order such that   $x_i>x_j$ whenever
$i\leq m$ and $j\geq m+1$. W.r.t. this order, $Q$ is  a Gr\"{o}bner
basis for $I=\ide{Q}$: indeed, take any $0\neq p\in I$  and assume
by contradiction that   $\lm(p)$ (the leading monomial of $p$) is
not divisible by any leading monomial $x_i$ in $Q$; that is,
$\lm(p)$ does not contain any $x_i$ with $i\leq m$. This would
imply, by definition of lex order, that $p$ does not contain any
such $x_i$, that is $p\in \reals[x_{m+1},...,x_N]$. Then for each
$v=(\mu_{m+1},...,\mu_N)$, we can consider $w=(g_1(v),...,g_m(v),
\mu_{m+1},....,\mu_N)\in \psi=\var(Q)$, implying $p(w)=p(v)=0$. In
conclusion, as $p(v)=0$ for each $v$,  $p$ is the zero polynomial,
contradicting the assumption.

Now let us   check that $\Id(\psi)=\ide{Q}=I$. Clearly
$\Id(\psi)\supseteq \ide{Q}$. On the other hand, consider any
  $p\in \Id(\psi)$ and let $p=q+r$, where $r=p\bmod Q$ and $q\in I$.
   By the above
  assumptions on $Q$ and by definition of remainder,
  no variable $x_i$ with $i\leq m$
   can occur in $r$, that is $r\in \reals[x_{m+1},...,x_N]$.
  Now assume by contradiction $r\neq 0$, so there is
   $v=(\mu_{m+1},...,\mu_N)$ such that $r(v)=a\neq 0$.
  Then  
    $w=(g_1(v),...,g_m(v), \mu_{m+1},....,\mu_N)\in \psi$, hence  $p(w)=0$; yet $p(w)=q(w)+r(v)
    =0+a\neq 0$, which is a contradiction. 
\end{proof}

\vsp
\newcommand{\cplx}{\mathbb{C}}
When computing $\Id(\psi)$ is not feasible, there is little alternative
to replacing it with some easy to compute ideal $I_0\subseteq\Id(\psi)$:
as discussed above, this preserves soundness of the approach, although
completeness is lost in general. A practical choice
 might be considering
$\sqrt{I}$, the \emph{complex} radical ideal of  $I$
\begin{eqnarray*}
\sqrt{I} & \defi & \{p\in \cplx[\xx]\,:\,p^m\in I \text{ for some } m>0\}\,
\end{eqnarray*}
where $\cplx$ denotes the complex field. By Hilbert's strong
Nullstellensatz \cite[Ch.4,§1.2,Th.6]{Cox}, in $\cplx[\xx]$ we
have
\begin{eqnarray*}
\sqrt{I} & = &\Id(\var_\cplx(I))
\end{eqnarray*}
where $\var_\cplx(I)=\{v\in \cplx^N\,:\,p(v)=0\text{ for each }p\in
I\}$ is the complex algebraic variety induced by $I$. As
$\var_\cplx(I)\supseteq \var(I)$, one has
\begin{eqnarray}
\sqrt{I} \cap\reals[\xx]& \subseteq& \Id(\psi)\,.\label{eq:I0}
\end{eqnarray}
Therefore, we can set $I_0=\sqrt{I} \cap\reals[\xx]$ and take as $G$
any Gr\"{o}bner basis of $\sqrt I$; note that, as $I\subseteq
\reals[\xx]$,   necessarily $G\subseteq \reals[\xx]$. The inclusion
\eqref{eq:I0} is in general {strict}. As an example, consider
$Q=\{x^2+1\}$, hence $\var(Q)=\es$: then trivially
$\Id(\psi)=\reals[\xx]$. On the other hand,
 $\var_\cplx(\{x^2+1\})=\{\iota,-\iota\}$
 hence $\sqrt{I}\cap  \reals[\xx]\neq \reals[\xx]$; for example
 $x\notin \sqrt{I}$.

The problem of computing a set of generators for the complex
radical of $I$ is well understood, and there exist well-known
algorithms to this purpose: in particular, those by Krick and Logar
\cite{KL} and by Laplagne \cite{Lap}. Although the worst-case
complexity of these methods is doubly exponential in the number of
variables, they often work   reasonably well and a number of
implementations are offered in  computer algebra systems, including
those in Singular's \verb"radical" library
\cite[Sect.D.4.14.7]{Singular}. We rely on this library in our implementation.


\ifmai
\begin{theorem}[soundness]\label{th:corr2} Consider any ideal $I\subseteq  \Id(\psi)$ and let
$G$ be   a Gr\"{o}bner basis of $I$ in the definition of the    \gdc\ algorithm.
The resulting algorithm is sound,
  in the sense that
the returned sets $(V,J)$   satisfy
\begin{itemize}
\item[(a)] $\pi[V]\subseteq  \pi[\field^n]\cap I_\psi$, hence for $\phi=\var(\pi[V])$ equation
\eqref{eq:safetypla}
is true;
\item[(b)] $J$ is an invariant ideal and $\pi[V]\subseteq J\subseteq I_\psi$.
\end{itemize}
\end{theorem}
\fi

\section{Experiments}\label{sec:experiments}
We report below the outcomes of three experiments we have conducted,
applying the $\post$ algorithm  
to   challenging
systems  taken   from  the literature. 
The execution times reported below are for an implementation
in Python under SageMath \cite{Sage}, running on a Core i5 machine\footnote{Code and examples available
at \url{https://github.com/micheleatunifi/postconditions/blob/master/Post.py}.}. Wherever possible, we compare our results with those
obtained by other authors.

\paragraph{Collision avoidance}\label{sub:collision}
We consider the two-aircraft dynamics used to study   collision
avoidance, discussed in many papers on hybrid systems
\cite{San10,Liu,Pla14}. The model is described by the   equations
below,  where the variables have the following meaning: $(x_1,x_2)$
and $(y_1,y_2)$ represent the Cartesian coordinates of aircraft 1
and 2, respectively; $(d_1,d_2)$ and $(e_1,e_2)$ their velocities;
applying the technique discussed in Remark \ref{rem:ghost}, we   also introduce the auxiliary
variables (parameters, hence 0 derivative) $\omega_1$ and
$\omega_2$, representing the angular velocities of the two aircrafts, and
$x_{10},x_{20}, y_{10},y_{20}$, $d_{10},d_{20},
e_{10},e_{20}$, representing   generic initial values of the
corresponding variables. Overall, the system's vector field $F_1$
consists of  18 polynomials over as many variables (including the auxiliary ones).
\vsp
\begin{equation*}\label{eq:Air}
\begin{array}{rclrclrclrcl}
\dot x_1 & \!=\! & d_1 \quad& \dot y_1 & \!\!=\!\! & e_1  \quad&
\dot d_1 & \!\!=\!\! & -\omega_1 d_2 \quad & \dot e_1 &
\!\!=\!\! & -\omega_2 e_2 \\
\dot x_2 & \!=\! & d_2 \quad& \dot y_2 & \!\!=\!\! & e_2  \quad&
\dot d_2 & \!\!=\!\! & -\omega_1 d_1 \quad & \dot e_2 &
\!\!=\!\! & -\omega_2 e_1 \,.
\vsp
\end{array}
\end{equation*}
We consider the precondition $\psi$ that assigns to each non constant variable
 the parameter corresponding to its (generic) initial value:
 $\psi=\var(\{x_1-x_{10}, x_2-x_{20},...\})$. Note that $G=\{x_1-x_{10},
 x_2-x_{20},...\}$ is a set of generators for $\Id(\psi)$, and in fact a
  Gr\"{o}bner basis w.r.t. the
 lexicographic order (Proposition \ref{prop:special}).
We then  consider a   complete template  $\pi$ of degree 2  over all
the  system's variables: $\pi$  is a linear combination of $n=190$
monomials that uses as many template parameters.
  We then run $\gdc(\psi,\pi)$, which returns, after $m=3$ iterations and about 16s,
 a pair $(V,J)$. The vector space $V$ corresponds to  a result template with 10
parameters, $\pi'=\sum_{i=1}^{10} a_i\cdot p_i$. The
 instances of $\pi'$ are therefore all and only the system's polynomial
invariants of degree $\leq 2$, starting from a fully generic
precondition (Theorem \ref{th:corr}(a)). These include
 all the polynomial invariants mentioned in \cite{San10,Liu},
and several new ones, like the following
\begin{equation*}
 -x_{10} d_{10} - x_{20} d_{20} + d_{10} x_1 + d_{20} x_2 + x_{10} d_1 - x_1 d_1 + x_{20} d_2 - x_2 d_2 \,.
\end{equation*}\noindent
Let   $\phi\defi \var(\pi'[\reals^n])$ be the  variety defined by
the result template $\pi'$. The invariant ideal $J$ returned by the
algorithm represents the  {weakest algebraic precondition}
$\chi\defi \var(J)$ such that $\chi\impl [F_1] \phi$: in other
words, the largest algebraic  precondition for which all   instances
of $\pi'$ are polynomial invariants (Theorem
\ref{th:instructive}(b)). Moreover, $\chi$ is also the  {weakest
algebraic invariant} included in $\phi$ (Theorem
\ref{th:instructive}(a)). A Gr\"{o}bner basis of $J$ consists of 12
polynomials that represent as many conservation laws of the system
(see  Appendix \ref{app:exp}).

\paragraph{Airplanes vertical motion}\label{sub:airplanes}
We consider the  6-th order longitudinal equations   that capture
 the vertical motion (climbing,
descending) of an airplane \cite[Chapter 5]{Airplanes}. The system is given
by the   equations   below, where the variables have the following meaning:   $u$ = axial
velocity, $w$ = vertical velocity, $x$ = range, $z$ = altitude, $q$ = pitch rate,
$\theta$ = pitch angle. We also have two equations encoding $\cos\theta$ and $\sin\theta$: note that, in the equations, these two are just variable names,   not transcendental functions themselves. Applying the technique discussed in Remark \ref{rem:ghost}, we   also introduce the following auxiliary
variables (parameters, hence 0 derivative):
  $g$ =
gravity acceleration; $X/m$, $Z/m$ and $M/I_{yy}$,   
  where $m$ is the mass of the airplane,  $M$ the aerodynamic and thrust moment w.r.t. the $y$ axis, $(X,Z)$  are the aerodynamics and thrust forces w.r.t. axis $x$ and $z$, and $I_{yy}$ is
  the second diagonal element of its inertia matrix (see also \cite{Airplanes,Pla14,Kong});    and    $u_0,w_0,x_0,z_0,q_0$,
standing for the generic initial values of the corresponding
variables.  Overall,  the system's vector field $F_2$ consists of
   17 polynomials over as many variables.
\begin{equation*}\label{eq:Air2}
\begin{array}{rclrclrclrclrcl}
\dot u & \!=\! & \frac X m - g\sin\theta-qw \quad& \dot z &
\!\!=\!\! & -u\sin\theta  +w\cos\theta  \quad& \dot w & \!\!=\!\! &
\frac Z m + g\cos\theta + qu    & \dot q &
\!\!=\!\! & \frac M {I_{yy}} \\
\dot x & \!=\! & u\cos\theta  + w\sin\theta  \quad& \dot \theta &
\!\!=\!\! & q \quad& \dot \cos\theta & \!\!=\!\! & -q\sin\theta
& \dot \sin\theta & \!\!=\!\! & q\cos\theta \,.
\end{array}
\end{equation*}
In order to discover interesting polynomial invariants, we consider
a complete template  $\pi$ of degree 2  over all the original
system's variables   plus two auxiliary variables, the latter representing
the monomials $qu$ and $qw$\footnote{We could dispense with these
auxiliary variables by considering a complete template of degree
3.}.  $\pi$  is a linear combination of $n=207$ monomials that uses
as many template parameters.  We apply the approach underpinned by Theorem
\ref{th:instructive}(b): we first pick up a precondition that
requires $\theta=0$ and assign
 (generic) initial values to the remaining variables,
  $\psi_0\defi \var(\{\theta,\sin\theta,\cos\theta-1,u-u_0,w-w_0,x-x_0,z-z_0,q-q_0\})$.
  Note that $G=\{\theta,\sin\theta,...\}$
   is a set of generators for $\Id(\psi)$, and in fact a
  Gr\"{o}bner basis w.r.t. the
 lexicographic order (Proposition \ref{prop:special}).
  We
then run $\gdc(\psi_0,\pi)$, which returns, after $m=8$ iterations
and about 26s, a pair $(V,J)$. The vector space $V$  corresponds to
the following result template.
{\small
\begin{eqnarray*}
\pi'  \; =\; \sum_{i=1}^4 a_i\cdot p_i & = &  a_1 \cdot\left(\cos^2\theta  + \sin^2\theta - 1\right)\quad
+\quad a_2 \cdot\left(-\frac 1 2 q^2 + \theta \frac M{I_{yy}} + \frac 1 2 q_0^2\right)
+\\
&&
 a_3 \cdot\left(u q \cos\theta + w q\sin\theta  - \frac X m\sin\theta  + \frac Z m \cos\theta   - x \frac M{I_{yy}} -  - \frac M{I_{yy}} x_0 + u_0 q_0 +\frac Z m\right)
 +\\
&&
 a_4\cdot \left(w q\cos\theta - u q \sin\theta - \theta g -  \frac X m\cos\theta   - \frac Z m \sin\theta   - z \frac M{I_{yy}}    -\frac M{I_{yy}} z_0 + w_0 q_0 + \frac X m\right)\,.
\end{eqnarray*}}
Let   $\phi\defi \var(\pi'[\reals^n])$  be the variety defined by
the result template $\pi'$. The invariant ideal $J$ returned by the
algorithm represents the  {weakest algebraic precondition}
$\chi\defi \var(J)$ such that $\chi\impl [F_2] \phi$: in other
words, the largest  algebraic  precondition for which all instances
of $\pi'$ are polynomial invariants (Theorem
\ref{th:instructive}(b)). Moreover, $\chi$ is also the  {weakest
algebraic invariant} included in $\phi$ (Theorem
\ref{th:instructive}(a)). A Gr\"{o}bner basis of
$J$ consists of 15 polynomials.
These findings generalize those in
\cite{Pla14,Kong}. In particular, one obtains the polynomial
invariants of \cite{Pla14,Kong} by letting $x_0=z_0=q_0=0$. By
comparison, \cite{Pla14} reports that their method spent 1 hour to
find a subset of all instances of  $\pi'$. The method in \cite{Kong}
reportedly takes $<1$s on this system, but again only finds a
subset\footnote{For instance, one should compare the polynomial
$\psi_3=q^2-2\frac {M\theta}{I_{yy}}$, which is part of the
invariant cluster in
 \cite{Kong},  with the polynomial  $p_2=-\frac 1 2 q^2 +
 \theta \frac M{I_{yy}} + \frac 1 2 q_0^2$  in the second summand of $\pi'$ above, which explicitly depends
on the initial condition $q_0$.} of instances of $\pi'$. Moreover,
it cannot infer the largest algebraic invariant implying the
discovered laws, as we do.

\paragraph{Kepler laws}\label{sub:kepler}
We want to show how the \gdc\ algorithm automatically discovers the
three Kepler's laws of planetary motion from Newton's law of
gravitation. A nice and self-contained explanation
 of these laws can be found in \cite{NewtKepl}. Newton's laws are expressed below
 in a system of polar coordinates  $(r,\theta)$ with the Sun at the origin.
The meaning of the variables is as follows:  $r$ is the planet's
distance
 from the origin; $\theta$ the angle from the positive horizontal semiaxis
to  the radius vector, measured   counterclockwise; $v_r$ and
$\omega$ the planet's radial and
 angular velocity, respectively; $u=1/r$ the distance reciprocal; for the purpose of
expressing the invariants of interest, the system also includes
 equations for
$\cos\theta$ and $\sin\theta$; moreover, we have constants (0
derivative variables)  $GM,a,e$ representing   the product
 of the gravitational constant  $G$ and the Sun's mass $M$,   the orbit's  major semiaxis
 and its eccentricity, respectively (see below). A
 few more dummy constants are used to encode positivity conditions. Overall,
 the system's vector field $F_3$ consists of 15 polynomials over as many variables.
\begin{equation}\label{eq:Newton}
\begin{array}{rclrclrclrclrclrclrcl}
\dot r & \!=\! & v_r \quad& \dot\theta & \!\!=\!\! & \omega \quad&
\dot v_r & \!\!=\!\! & - {GM}{u^2}+r\omega^2 \quad & \dot\omega &
\!\!=\!\! & -2 v_r \omega u & \quad \dot u &\!\!=\!\!& -u^2 v_r \\
&&&&&& \dot\cos\theta & \!\!=\!\! & -\omega\sin\theta \quad&
\dot\sin\theta & \!\!=\!\! & \omega\cos\theta\,.
\end{array}
\end{equation}
%
Because Kepler's laws concern closed  orbits\footnote{Note that non
closed, hyperbolic or parabolic, trajectories are also possible.},
we first seek for a precondition $\psi$ such that the planet's
motion is an ellipse of major semiaxis $a$ and eccentricity $e$. The
equation of such an ellipse  in polar coordinates, with one of the
foci coinciding with the origin (Sun) and the horizontal axis
passing through the ellipse's center, 
is $p_{\mathrm{ell}}=0$, where
\begin{eqnarray}\label{eq:ell}
p_{\mathrm{ell}} & \defi & r (1+  e \cos \theta)-  {a (1-e^2)}{}\,.
\end{eqnarray}\noindent
We consider a suitable $\psi_0$ that implies a unitary circular
orbit, which is an instance of $p_{\mathrm{ell}} $, and apply
Theorem \ref{th:instructive}(b): running $\gdc(\psi_0,\pi_1)$ for a
$\pi_1=a_1\cdot p_{\mathrm{ell}}$, we discover, in about 43s, the
largest (physically meaningful)   precondition $\psi$ implying
$p_{\mathrm{ell}}=0$. In particular, for
 $\omega_{\mathrm{in}} \defi   r^2\omega^2-GM\cdot  u\cdot (e+1)$ , we have $\psi= \var(P)$ where
\begin{eqnarray}\label{eq:psiK}
P & = & \left\{r-a(1-e), \theta,v_r, \omega_{\mathrm{in}}, u\cdot
r-1, \cos\theta-1, \sin\theta  \right\} \cup P_+\,.
\end{eqnarray}
Here the set $P_+$ encodes positivity    conditions on constants
($GM>0, a>0, 0\leq e<1$) and is omitted   for conciseness (further
details on the computation of $\psi$ and $\psi_0$ are given  in
Remark \ref{rem:psi0} below).


\newcommand{\pig}{\mathrm{Pi}}
We next consider   the complete polynomial  template $\pi_2$ built
out of monomials of degree $\leq 4$  on the variables
$GM,a,e,r,u,dA$, where $dA\defi  \frac 1 2 r^2\omega$ is an
auxiliary variable, representing the   areal velocity -- that is,
the first derivative of the area swept by the radius vector.   We
next run $\gdc(\psi,\pi_2)$, which returns, after $m=4$ iterations
and about 58s,
 a pair $(V',J')$.  The vector space $V'$ corresponds to  a result template
$\pi'_2 = a_1\cdot (ur-1)+a_2\cdot ( dA^2- a\cdot GM(1-e^2)/4) + R$,
where $R=\sum_{\ell=2}^{29} a_\ell \alpha_\ell$.  The term $ ur-1 $,
that is $u=1/r$, obtained by setting $a_1=1$ and the remaining template
parameters to 0, is another way of expressing Kepler's second law:
indeed, it implies that $\lie(dA)=- \omega r^2 u v_r + \omega r
v_r=0$, that is, that the areal velocity is constant. From Geometry,
we know that the ellipse's area is $A=\pig\  a^2\sqrt{1-e^2}$, where
$\pig=3.1415...$ denotes the transcendental mathematical constant. Since $dA$ is a
constant, the orbital period, expressed as a multiple of    $\pig$,
is $T\defi a^2\sqrt{1-e^2}/dA$. Therefore,   the second term in
$\pi'_2$, obtained by setting $a_2=1$ and the remaining template  parameters
to 0, can be read as saying that the square of the period, $T^2=
a^4(1-e^2)/dA^2$, is proportional to   $a^3$, the cube of the
semimajor axis: this is Kepler's third law. Any other summand  of
$\pi'_2$ is either a multiple of $ur-1$ or equivalent to the second
term, hence it gives no further information.

Let $\phi'=\var(\pi_2[\parv'])$. The invariant ideal $J'$ returned
by the algorithm represents the  {weakest algebraic
precondition} $\chi'\defi \var(J')$ such that $\chi'\impl [F_3]
\phi'$: in other words, the largest  algebraic  precondition
implying both the second and the third Kepler law  (Theorem
\ref{th:instructive}(b)).
 A
Gr\"{o}bner basis of the invariant  ideal $J'$ is $\{ur-1, dA^2-
a\cdot GM(1-e^2)/4\}$, hence  giving precisely the same information
as $\pi'_2$.

Rather than ``discovering'' the laws, it  is also possible to verify
them directly using  \gdc, that is  to check $\psi \impl[
F_3]\phi_i$, with: $\phi_1= \var(\{p_{\mathrm{ell}}\})$,
$\phi_2=\var(\{\lie(dA)\})$ and $\phi_3 = \var(\{T^2GM-4a^3\})$. The
running time  for these   checks is of about 45, 0.28 and 3s,
respectively.

\begin{remark}[on the computation of $\psi_0$ and $\psi$]\label{rem:psi0}{\em
Concerning the precondition $\psi_0$, we consider  a simple unitary
circular orbit, that is $p_{\mathrm{ell}}=0$  with $GM=a=1$ and
$e=0$. More precisely, considering as  $t=0$ to be
a time when the planet is on the positive semiaxis,
we let $\psi_0=\var(P_0)$ with $P_0=\{e,
\;a-1, \;GM-1, \;r-1, \;\theta,\; \;v_r,\;\omega-1,\, \; u-1 \}$ and
use the template $\pi_1=a_1\cdot p_{\mathrm{ell}}$. We  then run
$\gdc(\psi_0,\pi_1)$, which   returns  a pair $(V,J)$,  in $m=8$
iterations and about 43s. By Theorem  \ref{th:instructive}(b),
$\chi\defi \var(J)$ is the largest algebraic precondition implying
$p_{\mathrm{ell}}=0$. A set of generators for the invariant ideal $J$
consists of 9 polynomials (the Gr\"{o}bner basis is much larger, though). However, we   want to restrict
ourselves to physically meaningful
 initial conditions at time $t=0$, and to closed orbits.
Let   $J_0$ denote the ideal generated by the polynomial  encoding
of the following conditions:
   $v_r=\theta=\sin\theta=0$, $u\cdot r=\cos\theta=1$, $r=a (1-e)$ (from $p_{\mathrm{ell}}=0$),
 $dA=-r^2\omega/2$,   $GM>0$, $a>0$ and on $0\leq e <1$ (closed orbits).
We then define $\psi\defi \var(J+J_0)=\chi\cap\var(J_0)$. A small
set of polynomials representing $\psi$ is   obtained by computing a
Gr\"{o}bner basis $G$ of $\sqrt{J+J_0}$,
 the  complex  radical of $J+J_0$. From $G$,
via some simple manipulations, we compute the equivalent  set $P$ in
\eqref{eq:psiK}; that is, we have
   $\psi=\var(P)$.
}\end{remark}

\newcommand{\semi}{\mathrm{\mathbf{S}}}
\section{Application to continuous semialgebraic systems}\label{sec:semialg}
We illustrate an application of the $\post$ algorithm to the safety verification
of a class of continuous    systems, where both the  initial  set of states and
the set of  unsafe (`bad') states   are  \emph{semialgebraic} regions of $\reals^N$.  The family
 of semialgebraic sets, formally defined below,
 is  larger than the family  of algebraic sets and quite rich:
 for instance, in $\reals^3$ a  half-space, a disk, and  a ball are semialgebraic,
 but not algebraic sets. See \cite{Parrilo} for an introduction to semialgebraic sets and related
   techniques.
For the purpose of safety verification, the basic idea is that,
once we have obtained via $\post$
 an algebraic invariant for the system at hand, we can check if
  this invariant, as a region
of $\reals^N$,  intersects the specified  {unsafe} region: if not,
  the system is safe.
In pursuing this idea, we will systematically exploit the idea of auxiliary
 variables: 
we will have the obtained invariant be explicitly dependent
  on a set of parameters $\xx_0$,
  representing a generic initial condition for the given system.   Proposition \ref{prop:special} will
  guarantee   that a real radical for the initial set will be easy to compute.

A set $S\subseteq \reals^N$ is
\emph{(closed) basic  semialgebraic} if there are polynomials $g_1,...,g_m\in \reals[\xx]$ such
that $S=\{v\in \reals^N\,:\,g_1(v)\geq 0,...,g_m(v)\geq 0\}$,
written $S=\semi(\{g_1\geq 0,...,g_m\geq 0\})$\footnote{Note that an
equality $g_i=0$ can be coded up as a pair of inequalities $g_i\geq
0$ and $-g_i\geq 0$. Similarly, a strict inequality $g_i>0$ can be
coded up using an auxiliary slack variable $z$ as $g_i\cdot
z^2-1\geq 0$.}.
A (closed)  semialgebraic set is a finite union of basic semialgebraic sets.
In what follows, for the sake of simplicity we shall focus on  {basic} semialgebraic sets. It is very simple to extend the approach to general semialgebraic sets: this is outlined at the end of the section.
A \emph{ continuous (basic) semialgebraic system} is a triple
$SA=(F,X_0,X_U)$, composed by a   a polynomial vector field $F$, an
\emph{initial region} $X_0\subseteq \reals^N$ and an \emph{unsafe
region} $X_U\subseteq \reals^N$,  both of which are {basic semialgebraic}.
The system $SA$ is \emph{safe} if for each $v_0\in X_0$ there is no
$t\in D_{v_0}$ such that $\xx(t;v_0)\in X_U$.

Let us now introduce some
additional notation concerning auxiliary variables.
Let $F=(f_1,...,f_N)$ be a polynomial vector
field, with $f_i\in \reals[\xx]$. Let $\xx_0=(x_{01},...,x_{0N})$
be a vector of $N$ distinct variables, disjoint from $\xx$: we define the \emph{extended}
vector variables and vector field as
 $\hat\xx=(x_{01},...,x_{0N},x_1,...,x_N)$ and   $\hat
F=(0,...,0,f_1,...,f_N)$, respectively. Note that $\hat F$ is a vector field
$\reals^{2N}\rightarrow \reals^{2N}$, where the variables in $\xx_0$
represent generic constants. For $v,w\in \reals^N$, we will denote
by $(v,w)$ the element of $\reals^{2N}$ obtained by concatenating
$v$ and $w$. Finally, we will denote by $g[\xx_0/\xx]$ the
polynomial obtained from $g$ by replacing each variable $x_{i}$ with
$x_{i0}$, for $i=1,...,N$.

The following result gives a sufficient
algebraic condition for  safety of a continuous basic semialgebraic system.
Its intuitive interpretation is as follows.
In $\reals^{2N}$,
let $\psi$ be a precondition encoding
 that $\xx_0$ is the initial condition
for   $\xx$, and let
  $J$ be an invariant ideal representing a corresponding
postcondition, explicitly depending on $\xx_0$.
Hence, for any concrete instance of the initial conditions $\xx_0$,
 we obtain from $J$ a corresponding concrete postcondition.
If there is no solution of the set of (in)equations representing the
intersection of
 the initial region, of the postconditions and of the unsafe region,  then the system is safe.

\begin{theorem}[safety of semialgebraic
systems]\label{th:semialg} Let $SA=(F,X_0,X_U)$  be a basic semialgebraic
system, where $X_0=\semi(\{g_1\geq 0,...,g_m \geq 0\})$ and
$X_U=\semi(\{h_1 \geq 0,...,h_n \geq 0\})$ ($g_i,h_j\in
\reals[\xx]$).
Let $\psi=\var(\{x_i-x_{i0}:\,i=1,...,N\})\subseteq
\reals^{2N}$ 
and
let  $J=\ide{\{q_1,...,q_k\}}\subseteq \reals[\hat \xx]$ be an
invariant ideal for $\hat F$ such that $\var(J)\supseteq \psi$.
Assume the following polynomial system in
the variables $\hat \xx$
\begin{eqnarray}\label{eq:semi}
g_1[\xx_0/\xx]\geq 0,..., g_m[\xx_0/\xx]\geq 0,\;
 h_1\geq 0,....,h_n\geq 0,\; q_1=0,....,q_k=0
\end{eqnarray}
has no solution in $\reals^{2N}$. Then $SA$ is safe.
\end{theorem}
\begin{proof}
By contradiction, assume there are $w_0\in X_0$ and
$w_1=\xx(t_1;w_0)\in X_U$, for some $t_1\in D_{w_0}$. We will show
that $(w_0,w_1)\in \reals^{2N}$ is a solution of \eqref{eq:semi},
thus arriving at a contradiction. Indeed, by definition
$g_i[\xx_0/\xx](w_0,w_1)=g_i(w_0)\geq 0$ and $h_j(w_0,w_1)=h_j(w_1)\geq 0$ for
each $i=1,...,m$ and $j=1,....,n$. Consider now the trajectory of
$\hat F$ originating from $(w_0,w_0)$, that is $\hat
\xx(t;(w_0,w_0))$: note that, by definition of $\hat F$,  $\hat
\xx(t;(w_0,w_0))=( w_0\, ,\, \xx(t;w_0))$ for each $t\in D_{w_0}$. Now,
since $\var(J)\supseteq \psi$, we have $(w_0,w_0)\in \var(J)$, hence, by
$\hat F$-invariance of $J$, $\hat\xx(t;(w_0,w_0))\in \var(J)$ for
each $t\in D_{(w_0,w_0)}$ (Lemma \ref{lemma:invar}). In particular,
considering $t=t_1$, we have
$\hat\xx(t_1;(w_0,w_0))=(w_0\,,\,\xx(t_1;w_0))=(w_0,w_1)  \in \var(J)$.
But this means $q_i(w_0,w_1)=0$ for $i=1,...,k$. In conclusion,
$(w_0,w_1)$ is a solution of \eqref{eq:semi}.
\end{proof}

\vsp
There are two aspects of the previous result that are worthwhile commenting on.
First, checking that an algebraic system of
(in)equalities like \eqref{eq:semi}  is solvable is
decidable, although   NP-hard. One well-known and  effective technique
to establish insolvability is to rely on Positivstellensatz \cite{Stengle} and
Sum-of-Squares programming: this also provides easy to verify \emph{certificates}
of insolvability.  For the sake of completeness, we outline this
 technique in   Appendix \ref{app:Psatz}.

Second, the procedure resulting from the theorem is of course
 incomplete, and its precision depends on how rich the ideal $J$ is.
An invariant ideal $J$ satisfying the
hypotheses of the theorem can be obtained by running   $\post_{\hat F}(\psi,\pi)$
  with $\psi$
as specified in the statement of the theorem, and any template $\pi\in \Lin(\aa)[\hat\xx]$.
Indeed,   if $(V,J)=\post_{\hat
F}(\psi,\pi)$ (for some $V$), by Theorem \ref{th:corr}(b)
 $J$ is a $\hat F$-invariant ideal such
that $\var(J)\supseteq \psi$. The last point follows because
 $J\subseteq I_\psi$    implies that $q(w,w)=0$ for each $q\in J$ and
 $(w,w)\in \psi$. 
%
 Note that this is a case where a basis for the real radical $\Id(\psi)$ is trivial
  (Proposition \ref{prop:special}), hence  relative completeness holds.
Therefore, by   tuning the template $\pi$, one can in principle
 hope to obtain  a $J$   which is as precise as possible.
 The following example\footnote{SageMath/Python scripts for the examples in this section available at
 \url{https://github.com/micheleatunifi/postconditions/blob/master/Post.py}.} illustrates this theorem.

\begin{example}[3D Lotka-Volterra]\label{ex:3LV}{\em Consider the 3D
Lotka-Volterra system
defined  by $\xx=(x,y,z)$ and the vector field
$F=(xy-xz,yz-yx,zx-zy)$; see e.g.
\cite{San10,Kong}. Consider the basic semialgebraic system $SA=(F,X_0,X_U)$, where
$X_0=\semi(\{z=3,  (x-2)^2+(y-2)^2\leq 1.15^2\})$ (a disk)
and
$X_U=
\semi(\{(x-1/2)^2+(y-5)^2\leq 1.5^2\})$ (an infinite cylinder).
We wish to prove that $SA$ is safe.

Consider the extension of $F$, $\hat F$,
over the variables $\hat\xx=(x_0,y_0,z_0,x,y,z)$;
 let $\pi\in \Lin(\aa)[\hat\xx]$ be
a complete template
of degree 3.
Running $\post_{\hat F}(\psi,\pi)$
with $\psi=\var(\{x-x_0,y-y_0,z-z_0\})$, we get as a result
(after about 40s) a pair $(V,J)$
where  $J=\ide{\{q_1,q_2\}}$ and
\begin{eqnarray*}
q_1 & = & xzy-xz_0y_0-zz_0y_0-yz_0y_0
+z_0^2y_0+z_0y_0^2\\
q_2 & = & x_0+y_0+z_0-x-y-z\,.
\end{eqnarray*}
By the above discussion, $J$ is
a $\hat F$-invariant
ideal and $\var(J)\supseteq \psi$. For any instantiation of
 $x_0,y_0,z_0$ with real values,   $J$
represents a 1-dimensional variety in $\reals^3$, that is a curve,
 obtained as the intersection of
two surfaces. See Fig. \ref{fig:LV}(a). Any trajectory starting
in such a variety will remain in it.

\begin{figure}[t]
    \centering
    \begin{minipage}{0.45\textwidth}
        \centering
        \includegraphics[width=0.8\textwidth]{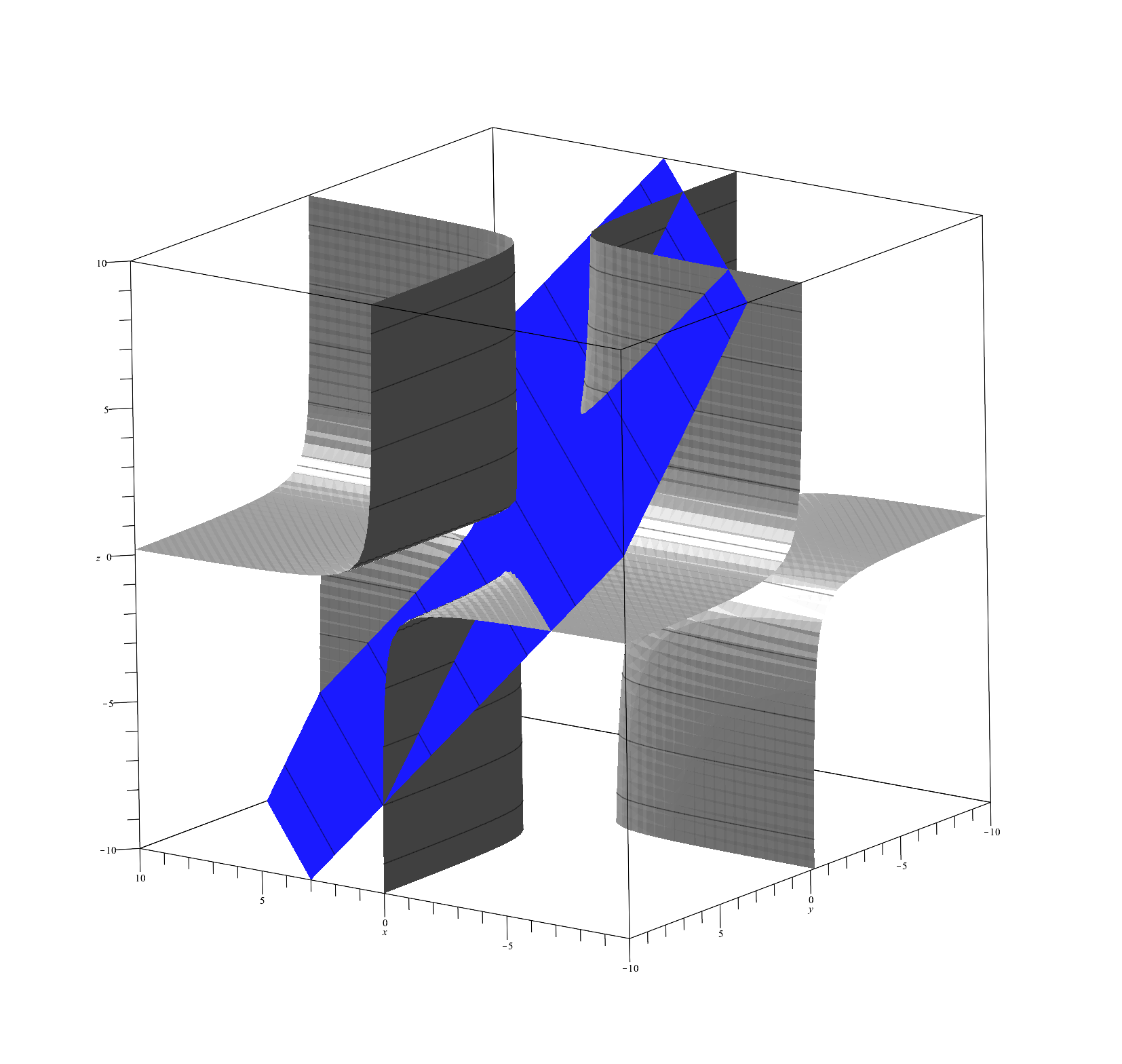} 

        \vspace*{-.3cm}
        {(a)}
    \end{minipage}\hfill
    \begin{minipage}{0.45\textwidth}
        \centering
        \includegraphics[width=0.8\textwidth]{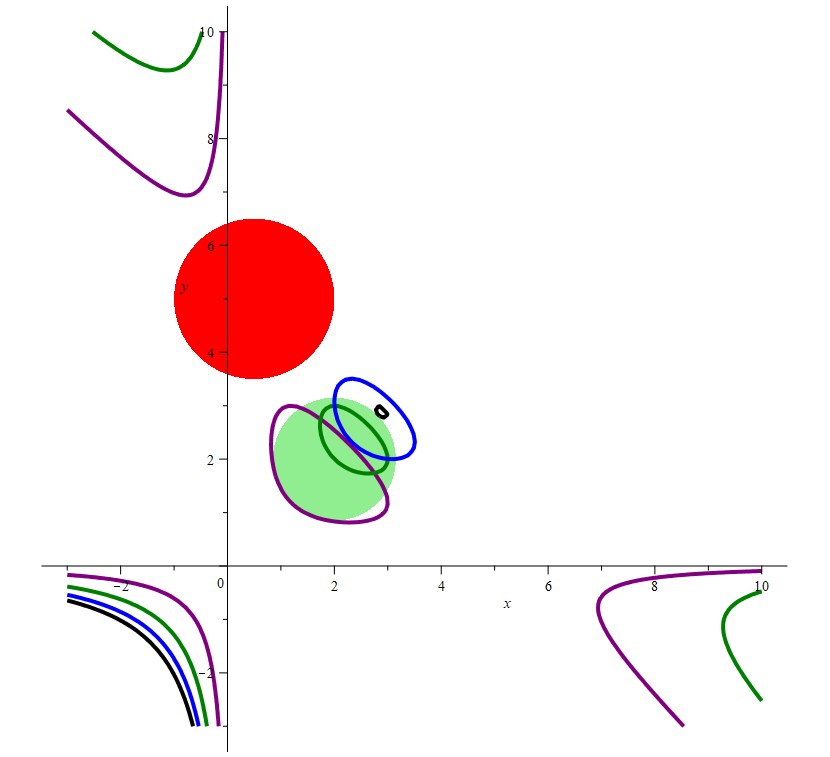} 

        {(b)
        }
    \end{minipage}
    \vspace*{-.2cm}
    \caption{With reference to the 3D Lotka-Volterra system in
         Example \ref{ex:3LV}: (a) surfaces in $\reals^3$ induced by the polynomials
         $q_1$ (grey) and $q_2$ (blue) in $J$, when instantiating
        $(x_0,y_0,z_0)$ to $(1,1,1)$; the corresponding algebraic invariant
         coincides with the intersection of the  two
        surfaces; (b) projection onto the $(x,y)$-plane of
        $X_0$ (green), of $X_U$ (red) and of
        the four algebraic invariants obtained from $J$ by instantiating $(x_0,y_0,z_0)$ to
          four points in $X_0$: $(2-c,2-c,3)$,   $(2,2,3)$, $(2,3.15,3)$ and $(2+c,2+c,3)$, where $c=1.15/\sqrt{2}$.}\label{fig:LV}
\end{figure}

\ifmai
\begin{figure}[t]
\begin{center}
  \includegraphics[width=.3\textwidth]{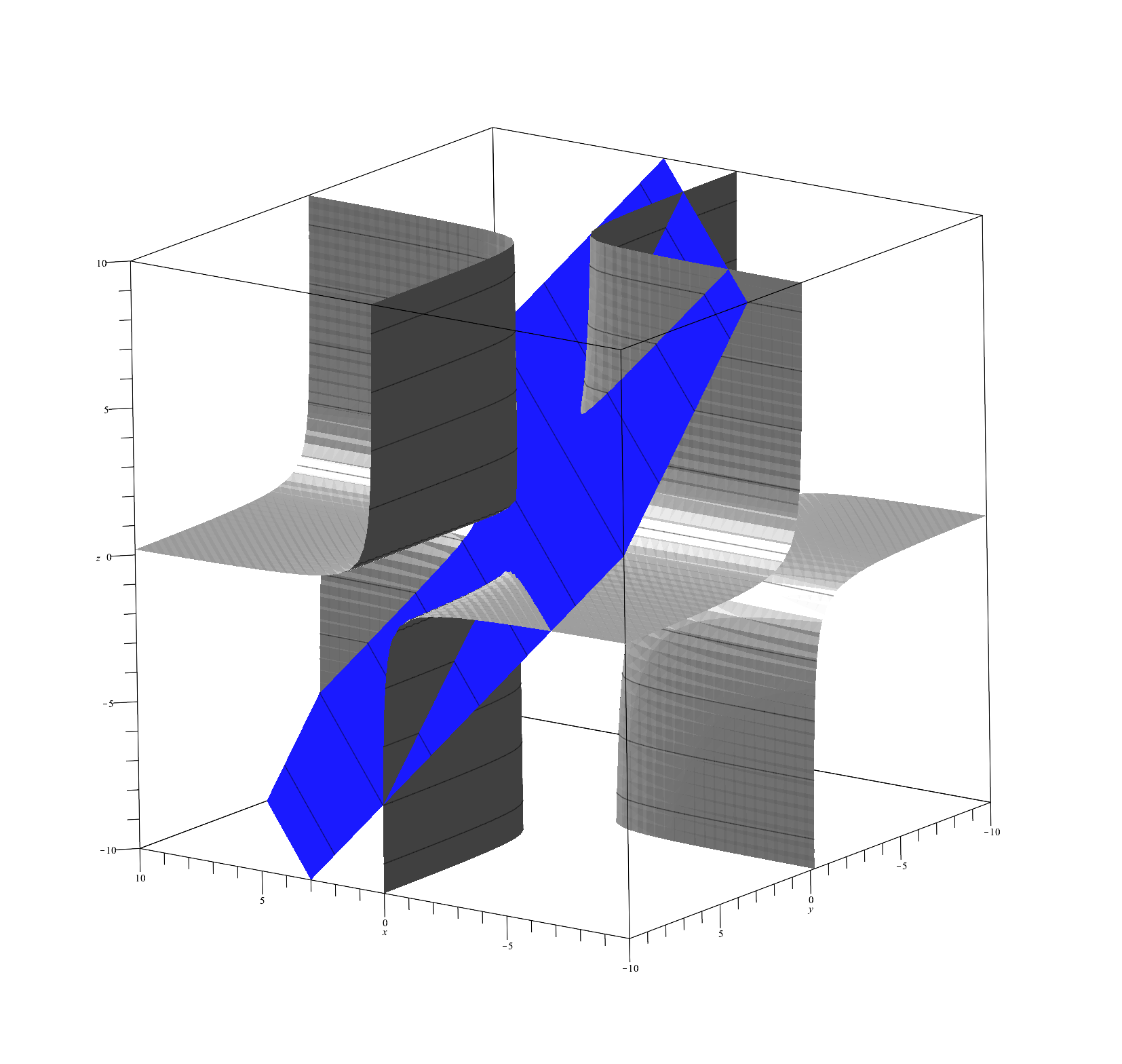}\hspace*{2cm}
  \caption{Bla}\label{fig:LV}
\end{center}
\end{figure}
\begin{figure}[t]
\begin{center}
  \includegraphics[width=.3\textwidth]{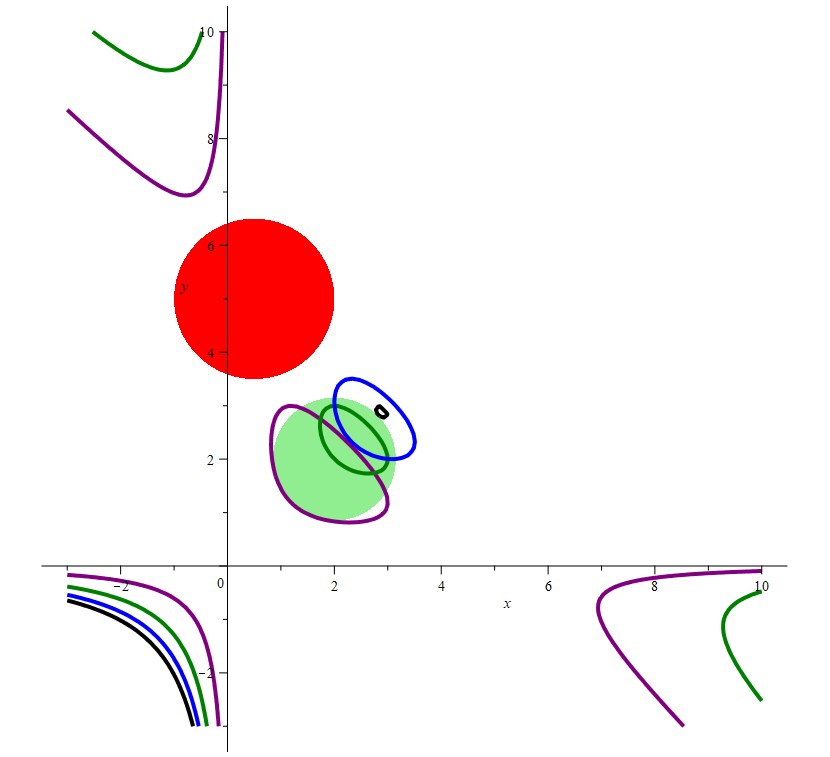}
  \caption{Bla}\label{fig:LV2}
\end{center}
\end{figure}
\fi

Fig. \ref{fig:LV}(b) shows the projection onto the $(x,y)$-plane of
$X_0$, of $X_U$
 and of  four  curves induced by $J$, when instantiating $(x_0,y_0,z_0)$ with four distinct
 points
 in  $X_0$.
 None of  those curves intersects the unsafe region,
  suggesting
 that the system might be safe. This can in fact be proven algebraically relying
 on \eqref{eq:semi}, which for the present case is equivalent to the following
 (we have eliminated the variable $z_0$ exploiting the equation $z_0=3$ for $X_0$):
 \begin{eqnarray*}
-(x_0-2)^2-(y_0-2)^2+ 1.15^2& \geq &  0\\
-(x-1/2)^2-(y-5)^2+1.5^2 & \geq &  0\\
xzy-3xy_0-3zy_0-3yy_0 +9y_0+3y_0^2 & =& 0\\
x_0+y_0+3-x-y-z &=& 0\,.
\end{eqnarray*}
This algebraic system is proven to have no solutions: we give the details,
 including
a certificate of insolvability, in   Appendix \ref{app:Psatz}. Therefore, by
virtue of Theorem \ref{th:semialg},  $SA$ is safe.}
\end{example}

Theorem \ref{th:semialg} admits the following  slight
generalization, that allows for a more flexible use
of auxiliary variables. Condition (a) requires that $\psi$ captures all points $(\vv_0,\vv_0)$
 for $\vv_0$ in the initial set.  We report a proof in Appendix \ref{app:Psatz}. 

\begin{theorem}\label{th:semialgext} Let $SA=(F,X_0,X_U)$  be a basic semialgebraic
system, and $\hat \xx$ and $\hat F$ be the extended variable  vector  and vector field,
like in Theorem  \ref{th:semialg}.
For $ r_1,...,r_N \in \reals[\xx_0]$, let
  $\psi=\var(\{x_i-r_{i}:\,i=1,...,N\})$.  Finally let
$J=\ide{\{q_1,...,q_k\}}\subseteq \reals[\hat \xx]$ be an
invariant ideal for $\hat F$ such that $\var(J)\supseteq \psi$.
Assume the following conditions hold true:
\begin{itemize}
\item[(a)] $(X_0\times X_0)\cap Id \subseteq \psi$
($Id$ = identity relation over $\reals^{2N}$);
\item[(b)] The   polynomial system in
the variables $\hat \xx$ in \eqref{eq:semi} has no solution in $\reals^{2N}$.
%
\end{itemize}
Then $SA$ is safe.
\end{theorem}

The slightly enhanced flexibility of Theorem \ref{th:semialgext} consists essentially in the fact that we can have a subset of the initial values
 fixed to constants.  This is illustrated in the following example.

\begin{figure}[b]
\begin{center}
\usetikzlibrary{calc,patterns,
                 decorations.pathmorphing,
                 decorations.markings}

\begin{tikzpicture}[scale=.9, every node/.style={scale=.9}]
\tikzstyle{spring}=[thick,decorate,decoration={zigzag,pre
length=0.3cm,post length=0.3cm,segment length=6}]
\tikzstyle{damper}=[thick,decoration={markings,
  mark connection node=dmp,
  mark=at position 0.5 with
  {
    \node (dmp) [thick,inner sep=0pt,transform shape,rotate=-90,minimum width=15pt,minimum height=3pt,draw=none] {};
    \draw [thick] ($(dmp.north east)+(2pt,0)$) -- (dmp.south east) -- (dmp.south west) -- ($(dmp.north west)+(2pt,0)$);
    \draw [thick] ($(dmp.north)+(0,-5pt)$) -- ($(dmp.north)+(0,5pt)$);
  }
}, decorate] \tikzstyle{ground}=[fill,pattern=north east
lines,draw=none,minimum width=0.75cm,minimum height=0.3cm,inner
sep=0pt,outer sep=0pt]

\node [draw, outer sep=0pt, thick] (M) [minimum width=1cm, minimum
height=1cm] {$m$}; \node [draw, outer sep=0pt, thick] (M2)
[minimum width=1cm, minimum height=1cm, xshift = 3cm] {$m$}; \draw
[thick,fill=white] (M2.south west) ++ (0.2cm,-0.125cm) circle
(0.125cm) (M2.south east) ++ (-0.2cm,-0.125cm) circle (0.125cm);

\node (ground) [ground,anchor=north,yshift=-0.28cm,minimum
width=10cm,xshift=2.53cm] at (M.south) {}; \draw (ground.north east)
-- (ground.north west); \draw (ground.south east) -- (ground.south
west); \draw (ground.north east) -- (ground.south east);

\node (fill) [ground,xshift=-0.15cm,minimum height = 0.3cm, minimum
width = 0.3cm] at (ground.west) {}; \draw (fill.north west) --
(fill.south west); \draw (fill.south west) -- (fill.south east);

\draw [spring] (M.east) -- (M2.west); \draw [thick, fill=white]
(M.south west) ++ (0.2cm,-0.125cm) circle (0.125cm)  (M.south east)
++ (-0.2cm,-0.125cm) circle (0.125cm);

\node (wall) [ground, rotate=-90, minimum width=3cm,anchor=south
east] at (fill.north west) {}; \draw (wall.north east) --
(wall.north west); \draw (wall.north west) -- (wall.south west);
\draw (wall.south west) -- (wall.south east);



\draw [spring] (wall.12) -- ($(M.north west)!(wall.12)!(M.south
west)$);

\node (k) at (wall.15) [xshift = 3.95cm,yshift=0.38cm] {$k$}; \node
(k2) at (wall.15) [xshift = 1.05cm,yshift=0.38cm] {$k$};



\end{tikzpicture}
\end{center}
\caption{A spring-mass system.}
 \label{fig:springmass}
\end{figure}

\begin{example}[coupled spring-mass system]\label{ex:springmass}{\em
A system consists of two identical springs of elastic constant $k$ and length $L$
 and two identical bodies of mass $m$,  connected in cascade: wall,
spring 1, mass 1, spring 2, mass 2. See Figure \ref{fig:springmass}.
This system is governed by the following equations, where
$x_1,x_2$ and $v_1,v_2$ denote, respectively,
 the bodies' positions and velocities
  on the horizontal axis with the origin fixed at the wall:
 \begin{eqnarray*}
 \dot{x}_1 & = & v_1\\
 \dot{v}_1 & = &  (k/m) (x_2-2x_1)\\
 \dot{x}_2 &=& v_2\\
 \dot{v}_2 & = &-(k/m)(x_2-x_1-L)\,.
 \end{eqnarray*}
Considering $k/m$ and $L$ as  0-derivative variables, that is constants,
we    let $\xx=(k/m,L,x_1,v_2,x_2,
v_2)$ and
$F=(0,0,v_1,k/m (x_2-2x_1),v_2,-k/m(x_2-x_1-L))$.
Consider the system $SA=(F,X_0,X_U)$ with initial set
 $X_0=\semi(\{k/m=1,L=1,(x_1-1/2)^2+(x_2-3/2)^2\leq 1/4, v_1=0,v_2=0\})$ and unsafe set
 $X_U=\semi(\{x_2-x_1\geq 2.17\})$.
 That is, we fix the value of both constants to 1
  and the initial velocities to 0, and  let the initial positions
    $(x_1,x_2)$ of the two masses
   vary in a disk of radius
  1/2 centered at   $(1/2,3/2)$.
  We then ask if the distance of the first mass from the second
  ever reaches or exceeds the value 2.17. Fig. \ref{fig:spring}(a), displaying
   the plots of $x_2(t;\vv_0)-x_1(t;\vv_0)$ for
  100 random initial conditions  $\vv_0\in X_0$, suggests that the system
  might indeed be safe. We now prove this fact. Note that, despite the linearity
  of the system,
  nonlinear invariants will be essential to prove safety.

\begin{figure}[t]
\begin{center}
  \includegraphics[width=.6\textwidth]{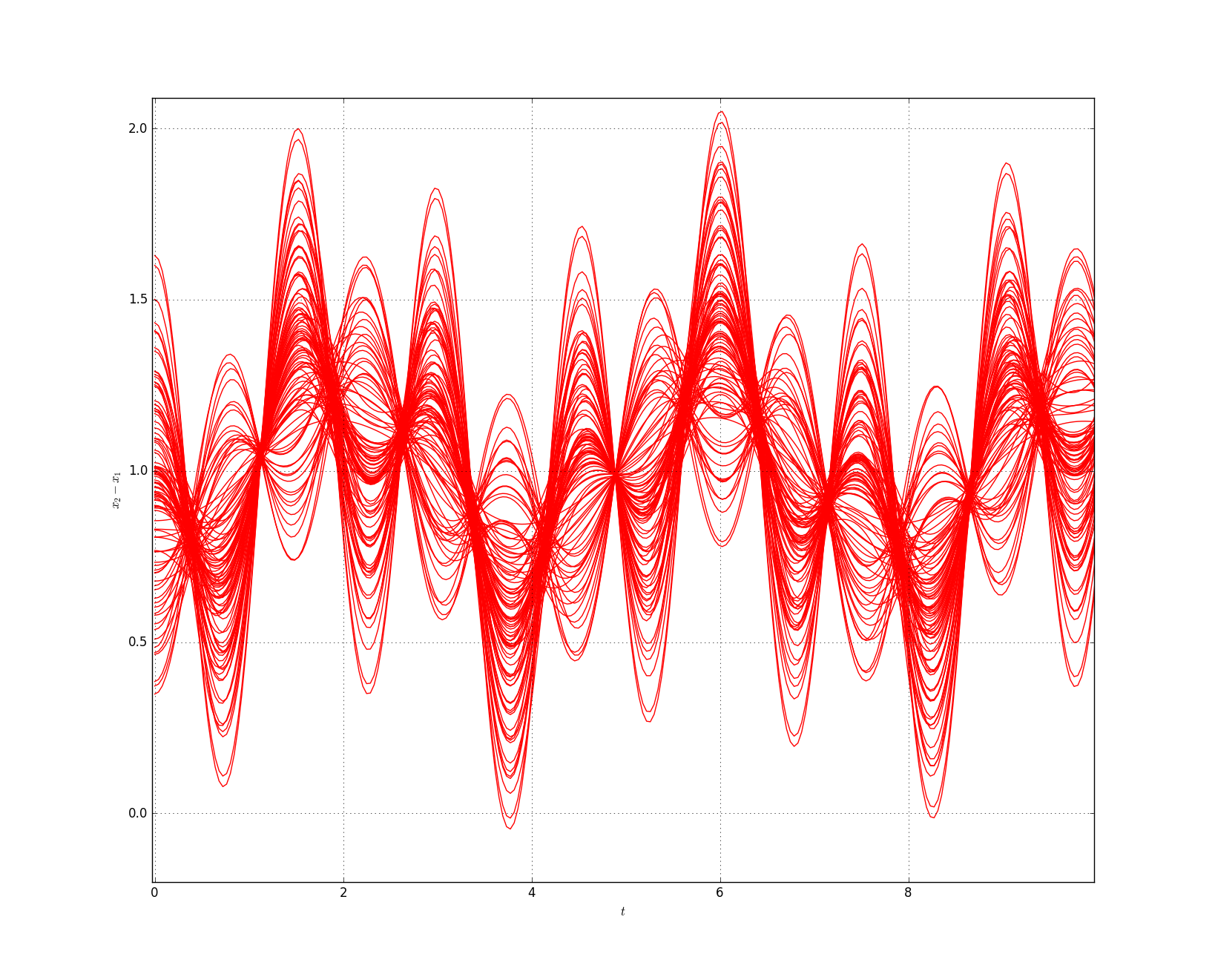}
  \caption{With reference to the coupled spring-mass system
  in Example \ref{ex:springmass}, plots of 100 trajectories
  $x_2(t)-x_1(t)$ starting from random points in $X_0$.}\label{fig:spring}
\end{center}
\end{figure}

As dictated by Theorem \ref{th:semialgext}, we
 consider an extended vector field $\hat F$ over the variables\footnote{In practice,
  it is superfluous   to introduce
 auxiliary copies of the constants $k/m, L$. The same is true for $v_1,v_2$:
 as the their initial values
  is fixed by $\psi$,   their auxiliary copies
   would    be anyway eliminated  from the final system \eqref{eq:semi}.} $\hat \xx$, then take a complete template
 $\pi$  of total degree
 3 over $\hat\xx$ and $\psi=\var(\{x_1-x_{10},x_2-x_{20},v_1,v_2\})$.  We obtain $(V,J)=\post_{\hat F}(\psi,\pi)$, which takes about
 60s, for some $V$ and   $J=\ide{\{q_1,q_2\}}$, where
{\small
  \begin{eqnarray*}
q_1 & = &  2/3 (k/m) L x_{10} + (k/m) x_{10}^2 - 4/3 (k/m) x_{10} x_{20} + 1/3 (k/m) x_{20}^2 - 2/3 (k/m) L x_{1} - (k/m) x_{1}^2 + 4/3 (k/m) x_{1} x_{2}-\\
 &&  1/3 (k/m) x_{2}^2 - 1/3 v_1^2 + 2/3 v_1 v_2\\
q_2 & = & 2/3 (k/m) L x_{10} + (k/m) L x_{20} - 1/3 (k/m) x_{10} x_{20} - 1/6 (k/m) x_{20}^2 - 2/3 (k/m) L x_{1} - (k/m) L x_{2} + 1/3 (k/m) x_{1} x_{2}\\
&&  +1/6 (k/m) x_{2}^2 + 1/6v_1^2 + 2/3v_1 v_2 + 1/2v_2^2\,.
 \end{eqnarray*}}%
\noindent
(A Gr\"{o}bner basis of $J$ consists of 4 polynomials whose lengthy description is omitted here). One can check that, once fixed $k/m=L=1$, for any instantiation
of the variables $x_{10}$ and $x_{20}$ in $J$, the dimension of the resulting variety is $\leq 2$,
implying it represents a good over-approximation the resulting system trajectory.
 Once $k/m$ and $L$ have been eliminated, the system \eqref{eq:semi} becomes
{\small
  \begin{eqnarray*}
  -(x_{10}-1/2)^2-(x_{20}-3/2)^2 +1/4& \geq & 0 \\
  x_2-x_1 -2.17& \geq & 0\\
  x_{10}^2 - 4/3 x_{10} x_{20} + 1/3 x_{20}^2 - x_{1}^2 + 4/3 x_{1} x_{2} - 1/3 x_{2}^2 - 1/3 v_1^2 + 2/3 v_1 v_2 + 2/3 x_{10} - 2/3 x_{1} & = & 0\\
  -1/3 x_{10} x_{20} - 1/6 x_{20}^2 + 1/3 x_{1} x_{2} + 1/6 x_{2}^2 + 1/6 v_1^2 + 2/3 v_1 v_2 + 1/2 v_2^2 + 2/3 x_{10} + x_{20} - 2/3 x_{1} - x_{2} & = & 0\,.
 \end{eqnarray*}}%
\noindent
This algebraic system is proven to have no solutions:   we give some details
in   Appendix \ref{app:Psatz}. Therefore, by
virtue of Theorem \ref{th:semialgext},  $SA$ is safe.
}
\end{example}

We end the section by remarking that   the verification method discussed so far easily extends to
 general semialgebraic sets. Indeed, let a continuous system $SA=(F,X_0,X_U)$
have $X_0=\bigcup_i A_{i}$ and  $X_U=\bigcup_j B_{j}$ as initial and unsafe set, respectively,
 where $A_i,B_j$ are basic semialgebraic sets. Then safety of $SA$ is equivalent to safety of each one of the basic
 systems $SA_{ij}=(F,A_i,B_j)$.


\section{Discussion}\label{sec:discussion}
We round off the technical development so far with a discussion on
what   novelties and benefits our approach actually delivers,
also in relation to existing work. We shall focus on three
aspects:  extended systems, preconditions and polynomial
invariants.

\paragraph{Extended systems and relational abstractions}
Our use of auxiliary variables and of extended systems $\hat F$ to
handle semialgebraic safety problems in Section \ref{sec:semialg}
appears to be strongly connected to \emph{relational abstractions}
of \ode s, introduced by Sankaranarayanan and Tiwari,  see
\cite[Def.5]{ST11}. Expressed in our notation, a (timeless)
relational abstraction for a system of \ode s $\dot{\xx}=F$ (with
$F\in \reals[\xx]^N$) is a relation $R\subseteq \reals^{N}\times
\reals^{N}$ such that whenever $(v_0,v_1)\in R$ then, for some time
$t_1\in \reals$, we have $v_1=\xx(t_1;v_0)$;
  recall that $\xx(t;v_0)$ denotes the unique solution of the
system from $v_0$. In other words, a relational abstraction
over-approximates the set of    possible transitions between states,
abstracting away from time. Seeing that relational abstractions are
infinite-state transition systems, one can apply to them known
verification techniques for   such systems, such as k-induction:
this is elaborated in \cite{ST11}. Now, consider   the statement of
our Theorem \ref{th:semialg}, and specifically the invariant ideal
$J$ mentioned there:  the condition $\var(J)\supseteq \psi$
precisely says that $\var(J)$ is a relational abstraction --- the
argument showing this fact is in the proof of the theorem itself. As
discussed in Section \ref{sec:semialg}, $J$ can be obtained by
running the \post\ algorithm. Therefore
\post\ can also be seen as a
relatively complete method for computing algebraic relational
abstractions.

\paragraph{Direct and generic approaches to preconditions}
Throughout the preceding sections, we have  been emphasizing the
important role played in our approach by preconditions $\psi$. Summing up,
there are
  two distinct sensible  ways one can make use of this information in conjunction
  with the  \post\ algorithm.
  \begin{enumerate}
  \item
\emph{Direct approach}: use $\psi$  {directly} as an argument of $\post$,  and obtain an
algebraic invariant $\var(J)$ containing the precondition from
$(V,J)=\post_F(\psi,\pi)$.
 \item 
 \emph{Generic approach}: introduce new variables $x_0,y_0,...$
 and $\psi_0=\var(\{x-x_0,y-y_0,...\})$ to represent arbitrary initial conditions
and compute a set of  parametrized invariant ideals $J$ by running $\post_F(\psi_0,\pi)$.
Then syntactically replace
  $x,y,..$ by $x_0, y_0...$   in the description of the precondition $\psi$
   and add the resulting system of equations,
 representing the constraints on the initial conditions, to $J$.
\end{enumerate}
The generic approach is basically
how we have handled semialgebraic problems  in Section \ref{sec:semialg}.
 It is instructive to compare the two approaches on one and the same example.
  Let us reconsider the 3D Lotka-Volterra system
 of Example \ref{ex:3LV}, where we have already applied the generic approach.
 Let us now apply the direct approach to this system.
 Introducing an extra slack variable $e$, and extending $F$
 with $\dot e =0$,
 we can define the set of initial conditions as the algebraic
 variety $\psi\defi \var(\{z-3,  (x-2)^2+(y-2)^2+e^2- 1.15^2\})$.
Running $\post$ with $\psi$ and
 a  complete template $\pi$ of degree 3,  after $m=3$ iterations and $<1$s
 we obtain  $(V,J)=\post_{F}(\psi,\pi)$. The invariant ideal $J$
 is generated by the   polynomial
 \begin{eqnarray*}
q&=&xyz - (3/2)(e^2 + x^2 + y^2+ z^2)- 3 (x y  + x z + y z)  + 15 (x +   y +   z) - 33213/800
\end{eqnarray*}
As $J\subseteq I_\psi$ (Theorem \ref{th:corr}(b)), we have $\psi   \impl    [F]\;\var(J)$,
that is, the the system's
trajectories starting from $\psi$ never leave $\var(J)$.
To prove safety, it is then enough to check that
  $\var(J)\cap X_U=\es$. Ultimately, this reduces to proving that the algebraic
system $S=\{q=0,(x-1/2)^2+(y-5)^2\leq 1.5^2\}$
has no real solutions, which can be  readily checked  to be the case via the same SOS programming
 technique\footnote{Or, in this simple case, by just feeding $S$
  to a computer algebra system.}
  illustrated
in Section \ref{sec:semialg}. In this case, the direct approach
 leads to  a gain
in terms of execution time. This gain is mainly
 imputable to the   smaller number of variables. Moreover, in this and other examples,
 starting from a   specific   precondition
 apparently leads  to  a faster convergence of $\post$.
However,     the direct approach often fails, as the smallest
algebraic invariant including $\psi$  is simply not precise enough (too large) to establish safety.
This is the case with the spring-mass system of
Example \ref{ex:springmass}, as there is  no
 nontrivial algebraic invariant   including   the
given initial set.

The generic approach appears to be more effective at chasing invariants: informally,
   not just one, but  a whole family of polynomial invariants can be captured at once, provided    the family
can be described parametrically
with respect to a generic point  $\xx_0$.
Note however that when the
   dependence from $\xx_0$ is not parametric,
 and invariance only holds for certain values   $\xx_0$ may take on,
 the generic approach  fails. For  instance, the system $\{\dot x = 2x^2y-x,\,\dot y = 2xy^2+y\}$ has both $x$ and $y$ as   polynomial invariants for $(0,0)$: indeed e.g. $\dot x=(2xy-1)x$. However, neither
  is found  using the generic approach.  Vice-versa, both
   invariants  are found using the direct approach,  with the precondition $\psi=\{(0,0)\}$, the single point where the lines
   $x=0$ and $y=0$ intersect.


\paragraph{First-, second- and   higher order integrals}
Consider an invariant ideal $J$, possibly found using   \post,
and any $q\in J$. If $\lie(q)\in \ide{\{q\}}$,
 $q$ is called a \emph{Darboux} polynomial.
An important   special case of Darboux polynomials is when $\lie(q)=0$, which, in the language of
mathematical
physics, makes
 $q$     a so-called
 \emph{first integral} of the system. In our terminology,
 up the to an additive constant, first integrals coincide
 with      polynomials that are invariants  \emph{for each} $\vv_0\in \reals^N$ in the sense
 of Definition \ref{def:polyinv}.
 A Darboux polynomial which is not a first integral
 is also known as a    \emph{second integral} in mathematical physics.
  Polynomials of an invariant ideal
  other than Darboux
  are collectively designated as higher order integrals; see \cite[Ch.2]{Goriel}.
%
Integrals of any order are  important in applications, as
  the knowledge of any one of them decreases
    the number of \emph{degrees of freedom} of the system --- in our terminology, the dimension
    of the smallest algebraic variety including the system's trajectories.
Finding
all polynomial \emph{first} integrals, up to a given degree,  can be done quite
easily using just templates and linear algebra: indeed,
 the constraints for the derivative
of the parametric template to be the   zero polynomial are always linear.
Things get harder when one needs to compute   second integrals
and beyond.  $\post$ seamlessly computes
 polynomial integrals of any order up to a given degree.  It is   interesting to inspect the
invariant ideals of
  the
examples seen so far, and
check how many polynomials, in the corresponding sets of generators,
 are   computationally
interesting, i.e. \emph{not} first integrals. Table \ref{tab:order}
displays how many integrals of each type have been found for each
example by $\post$. In all considered cases, but   Kepler's First
Law, the set of generators is also a Gr\"{o}bner basis. Appendix
\ref{app:exp} presents the details of a higher order integral  from
the Collision Avoidance example.

\begin{table}[t]
\begin{tabular}{|l|c|c|c|c|c|}
\hline
\multicolumn{1}{|c|}{\textbf{Example}}& {\bf \# 1st integrals} & {\bf \# 2nd integrals} & {\bf \#  higher order integrals } \\
\hline
Collision Avoidance & 6 & 0 & 6\\
\hline
Airplane Vertical Motion & 3 & 0 & 12\\
\hline
Kepler 1st Law & 0 & 0 & 9\\
\hline
Kepler 2nd \& 3rd Laws & 0 & 1 & 1\\
\hline
3D Lotka-Volterra & 2 & 0 & 0\\
\hline
Spring-mass System & 2 & 0 & 2\\
\hline
\end{tabular}
\centering
  \caption{Classification of   polynomials in  invariant ideals.}
\label{tab:order}
\end{table}

\section{Conclusion, further and related work}\label{sec:concl}
We have provided complete algorithms to compute 
relativized   strongest postconditions for systems of polynomial \ode s.
These algorithms can be used to
check safety assertions, to discover complete sets of polynomial invariants that fit a given template,
and to compute
largest algebraic varieties of initial conditions making
 given polynomial invariants true (weakest preconditions).
Effectiveness of the algorithms has been demonstrated
on  nontrivial systems, including semialgebraic ones.


Our previous work \cite{Fossacs17} deals with simple initial values
problems, where the precondition $\psi$ always consists of a
singleton. This restriction prevents one from dealing with the most
interesting continuous systems, such as semialgebraic systems. In particular,
 the concepts of weakest precondition is absent in \cite{Fossacs17}. The present paper lifts
 the algorithm of \cite{Fossacs17} to general (semi)algebraic systems. This requires
 considering a general algebraic precondition, described by a Gr\"{o}bner basis $G$,
  rather than a singleton. Among the new technical ingredients
necessary to   make this extension work the following two are crucial: (a) the property that, under
 suitable conditions,  reduction modulo $G$ and
  substitution commute with
 each other (Lemma \ref{lemma:goebner}); (b) the use of auxiliary variables and generic initial
 conditions to   circumvent the real radical problem (see e.g. Proposition \ref{prop:special} and Theorem
 \ref{th:semialg}).
The method introduced in \cite{Fossacs17} has its roots in a line of
research concerning weighted automata, bisimulation  and Markov
chains \cite{Bor09,IC12,Bor15}.  Also   related to the
present paper is \cite{HSCC18}, where we apply the notion of
invariant ideal to the construction of linear abstractions of
continuous systems.

The study
of the safety of hybrid systems can be shown to
reduce constructively to the problem of generating invariants for
their differential equations \cite{Pla12}. Many authors have therefore focused
on  the  effective generation of invariants of  a special type. For example,
 Tiwari  and Khanna consider invariant
generation based on  \emph{syzygies} and Gr\"{o}bner basis
\cite{Tiwa04}. Sankaranarayanan \cite{San10} characterizes greatest
invariants in terms of a descending chains of ideals. This iteration
does not always converge, thus a relaxation in terms of
bounded-degree \emph{pseudoideals} is considered: the resulting
algorithm always converges, because pseudoideals form basically a
descending chain of finite-dimensional vector spaces, and returns an
invariant ideal, although with no guarantee of maximality
\cite[Th.4.1]{San10}. By contrast, the convergence of our algorithm
{\sc Post} is essentially based  on the stabilization of ascending
chains of ideals, with completeness guarantees. Matringe et al.
encode invariants
 constraints using symbolic matrices \cite{Matr}.

Our work is closely related to Ghorbal and Platzer's \cite{Pla14},
which gives a  complete characterization  of what it means for an
algebraic set to be invariant for a polynomial \ode. It is
interesting to contrast the completeness statements in \cite{Pla14}
and of that in the present paper with one
another. \cite{Pla14} presents a method that, given a polynomial
template and an integer $M\geq 1$,
   determines the largest subspace of   template
instantiations  under which a length $M$ chain of Lie derivatives
forms an invariant. Any invariant can be reduced to this form, for suitably large $M$.
In contrast, our Theorem \ref{th:corr}, given a
template \emph{and} an algebraic variety   $\psi$,
determines the largest subspace of the template instantiations that
are polynomial invariants for trajectories
 starting from $\psi$; moreover, it determines, via $J$, the largest
invariant variety that  includes $\psi$. Neither of these two
statements is stronger or more general than the other. From our point
of view, taking the initial set explicitly into account, as we do,
has some advantages. First,   invariants $J$ returned by our method
can be  made explicitly dependent  on initial conditions, via auxiliary variables, such
as $x_0,z_0,w_0,u_0,q_0$ in the longitudinal airplane motion. As
such, these invariants can  be used \emph{directly}  within semialgebraic
verification methods based on Positivstellensatz:  as seen in Section \ref{sec:semialg}, this amounts to
 proving
the unsatisfiability of a set of polynomial (in)equations, also involving
the auxiliary variables, corresponding to the initial set, to the polynomial invariants, and to the
unsafe set. 
Second, knowing the precondition
$\psi$, we are in effect confining  ourselves to a subset of the algebraic invariants, those
 that involve points in $\psi$: this might explain
the observed  gain in efficiency
 --- practically speaking, as the worst-case complexity is   left unchanged. This gain
is reflected in the execution times of \cite{Pla14} and of our
algorithm, for the examples   reported in Section
\ref{sec:experiments}. As a more general remark, we  note that the computational
ingredients of \cite{Pla14}, such as minimization of the rank of a
symbolic matrix (also employed in  \cite{Matr}),  are quite
different from ours.  In the future, we would like to experimentally
compare these two approaches  on a more systematic basis than what
we have done in the present paper.

The recent work of Kong et al. \cite{Kong} considers  generation of
invariant  clusters, again based on templates. Nonlinear constraints
on    template parameters are resolved via symbolic computation;
safety for semialgebraic systems is reduced, via Positivstellensatz, to Sum-of-Squares (SOS)
programming.
In terms of   effectiveness, this method appears to
considerably improve      previous techniques. Kong et al.'s approach has strong similarities
with our method
for semialgebraic systems. Rather than relying on clusters, we generate families of   invariants via the introduction of
   auxiliary variables
$\xx_0$, denoting arbitrary points in the corresponding varieties. Differently from our approach, though, \cite{Kong} does not offer completeness
guarantees
 in our sense.
  In particular, the method of \cite{Kong} only sometimes works
  with chains of ideals that stabilize after one step, that is
  Darboux polynomials.
  On the other hand,
compared to theirs, our approach appears to be slower. It would be
interesting to investigate if the more general invariants $J$
returned by our algorithm could be fruitfully employed in the
approach of \cite{Kong}.

Ideas from Algebraic Geometry have been fruitfully applied also in
Program Analysis. Relevant to our work is M\"{u}ller-Olm and Seidl's
\cite{Muller}, where  an algorithm to compute all    polynomial
invariants up to a given degree of an imperative program  is
provided. Similarly to what we do, they  reduce the core   problem
to a linear algebraic one. However, since the setting in
\cite{Muller} is discrete rather than continuous, the  techniques
employed there are otherwise quite different, mainly because: (a)
the construction of the ideal chain is driven by the program's
operational semantics, rather than by Lie derivatives; (b) only the
polynomial invariants satisfied by \emph{all} initial program states
are considered, which in a continuous setting would mostly lead to
the trivial strongest postcondition. A perhaps more crucial
difference is that, when computing  invariants,  \cite{Muller}
regards templates essentially as polynomials  in
$\reals[\mathbf{x},\mathbf{a}]$, rather than explicitly factoring
out the template parameters $\mathbf{a}$. This can make the involved
computations less efficient, as known algorithms for computing
Gr\"{o}bner bases have a complexity which is exponential in the
number of variables.


Most of the material in this paper has been  extended and revised from the conference paper
\cite{Sofsem18}. With respect to \cite{Sofsem18}, here we
  include the following additional material:
proofs, the discussion on the expressive power of auxiliary
variables in Section \ref{sec:problem2}, the algorithmic presentation
and the discussion
on radical ideals in Section \ref{sec:compasp},   the
examples about   collision avoidance and Kepler laws    in Section
\ref{sec:experiments}, the extension to semialgebraic continuous systems in Section \ref{sec:semialg},
and an extended and revised discussion of related works in Section \ref{sec:discussion} and in the present section.

\noindent
\paragraph{Acknowledgment} The author has benefited from stimulating
discussions with Khalil Ghorbal.

\iffull
\newpage
\appendix
\section{A simple algorithm for weakest preconditions}\label{app:pre}
Fix a vector field $F$. Let $\phi=\var(P)$ be a user specified
postcondition, with $P\subseteq \reals[\xx]$ a finite set of
polynomials. We define inductively the sets $P_j$, $j\geq 0$, as
follows: $P_0\defi P$ and $P_{j+1}= \lie(P_j)$. For $j\geq 0$, we
let
\begin{eqnarray}\label{alg:propost}
I_j\defi \ide{\cup_{i=0}^j P_i}\,.
\end{eqnarray}
Let $m$ be the least integer such that $I_m=I_{m+1}$, which must exist
as $I_0\subseteq I_1\subseteq\cdots$ forms an infinite ascending
chains of ideals that must eventually stabilize. We let
$\pre(\phi)\defi I_m$. Note that the termination condition reduces
to checking   equality between two ideals, which    can be
effectively done (Section \ref{sec:prel}).

\begin{theorem}[correctness and completeness of \pre]\label{th:userpost}
Let $\phi$ be an algebraic variety and $I=\pre(\phi)$. Then
$\var(I)=\psi_\phi$.
\end{theorem}
\begin{proof}
Let $\chi\defi\var(I)$. It is easy to check that
$I=\ide{\{p^{(j)}\,:\,j\geq 0\text{ and }p\in P\}}$ and that
  $I$
is   an invariant ideal. By Lemma \ref{lemma:invar} then $\chi$ is
an algebraic invariant of $F$, that is $\chi\impl[F]\,\chi$.
Moreover, as $P\subseteq I$, $\phi\supseteq \chi$, hence
$\chi\impl[F]\,\phi$. This shows that $\chi$ is a valid precondition
of $\phi$. We now show  that it is actually the largest. Consider
any $\psi$ such that $\psi\impl[F]\,\phi$ and any $\vv_0\in \psi$.
This means that, for each $p\in I$, $p$ is a polynomial invariant
for $\vv_0$. That is (Lemma \ref{lemma:lie}), for each $p\in I$ and
$j\geq 0$, $p^{(j)}(\vv_0)=0$. Therefore,  $\vv_0\in
\var\left(\,\{p^{(j)}\,:\,j\geq 0\text{ and }p\in
P\}\,\right)=\var(I)=\chi$.
\end{proof}

\section{Proofs of Section \ref{sec:problem2}}\label{app:proofs}

\begin{proof_of}{Lemma \ref{lemma:Ipsi}}
The function $p(t;\vv_0)=p(\xx(t;\vv_0)$ of the real variable $t$ is
analytic in a neighborhood of 0. Hence it is identically 0
  if and only if all of its derivatives, $\frac{{d^j}}{ {d} t^j}p(t;\vv_0)$ for $j\geq 0$,
  vanish at
  $t=0$.
Then \eqref{eq:liederj} and \eqref{eq:initial}  establish the
result.
\end{proof_of}
\vsp

In the proof of the next lemma, we shall rely on the notion of monomial   ordering   \cite[Ch.2,§2]{Cox}, which we introduce below.

\newcommand{\multideg}{\mathrm{multideg}}\noindent
\begin{definition}[monomial ordering]
Let $k\geq 1$. A \emph{monomial ordering} $>$ on $\nat^{k}$  is a relation $>$ on $\nat^{k}$ such that: (i) $>$ is a total order on $\nat^{k}$; (ii) whenever $\aalpha>\bbeta$ and $\ggamma\in \nat^k$ then $\aalpha+\ggamma>\bbeta+\ggamma$; (iii) $>$ is a well-ordering on $\nat^k$: every nonempty subset of $\nat^k$ has a smallest element under $>$.

Let $\zz=(z_1,...,z_k)$ be $k$ distinct indeterminates. For $\aalpha\in\nat^k$, let
  $\zz^\aalpha$ denote the monomial $z_1^{\alpha_1} \cdots z_k^{\alpha_k}$. The monomial order $>$ is lifted to the set of monomials over $\zz$  by letting $\zz^\aalpha > \zz^\bbeta$ iff $\aalpha > \bbeta$. In $\tau=\zz^\aalpha$, $\aalpha$ is the \emph{multidegree} of $\tau$,
denoted $\multideg(\tau)$. The \emph{leading term} of a polynomial $p\in \reals[\zz]$ is the monomial $\tau$ appearing in $p$ of highest multidegree; then $\multideg(p)=\multideg(\tau)$.
\end{definition}

An example of monomial ordering is the  \emph{lexicographic order}:
at the level of multidegrees,  for  $\aalpha= (\alpha_1 ,...,\alpha_k )$ and  $\bbeta= (\beta_1 ,...,\beta_k )$
one lets $\aalpha >_{\mathrm{lex}} \bbeta$  if the leftmost nonzero entry of the vector
difference
$(\aalpha-\bbeta)\in \mathbb{Z}^k$ is positive. One lets $\zz^\alpha >_{\mathrm{lex}} \zz^\bbeta$ iff   $\aalpha >_{\mathrm{lex}} \bbeta$. For instance, for $\zz=(a_1,a_2,x,y)$, one has $a_1xy>_{\mathrm{lex}}x^3y^2$.

\vsp
\noindent
\begin{proof_of}{Lemma \ref{lemma:goebner}}
Consider the template $\pi$ as an element of $\reals[\zz]$, for $\zz\defi\aa,\xx=(a_1,...,a_n,x_1,...,x_N)$.
Note that, by
  our
choice of a lexicographic order where $a_i>x_j$ for any $i,j$,  we have for example,
  $\multideg(a_i^2 x_j)>\multideg(a_i x_j^k)$,
whatever $k\geq 0$.

Now let  $r=\pi\bmod G$, where again $r\in \reals[ \aa,\xx]$. We first
prove that $r$ is a template as well, that is, the template parameters $a_i$ can occur only linearly in $r$.
By the  properties of multivariate division  \cite[Ch.2,§3,Th.3]{Cox},
  there is a   $q=\sum_\ell h_\ell g_\ell$, with
$h_\ell\in \reals[\aa,\xx]$ and $g_\ell\in G$, such that
\begin{eqnarray}\label{eq:division}
 \pi & = & q+r\,.
\end{eqnarray}
Moreover, again by the same result: (a) $\multideg(\pi)\geq \multideg(q)$; (b) $r$ is a linear combination of
 monomials, none of which is
divisible
by the leading term of any polynomial in $G$.
Assume by contradiction
 there is  in $r$   a summand  $\mu \tau$ ($0\neq \mu\in \reals$), where
a template parameter $a_i$ occurs in the monomial $\tau$ with
 a degree $>1$. By the linearity of $\pi$ in the template parameters in
$\aa$ and by \eqref{eq:division}, we deduce that $-\mu\tau$ must be a summand
 of $q$ (seen as a linear combination of monomials),
 so that
the two terms can  cancel   each other. We deduce
that $\multideg(q)\geq \multideg(\tau)$. Hence, by (a) above and transitivity,
$\multideg(\pi)\geq \multideg(\tau)$. But this is impossible: indeed,  by the chosen lexicographic order,
we must have  $\multideg(\pi)< \multideg(\tau)$, because $\pi$ is linear in all the template parameters in $\aa$, whereas in $\tau$
the degree of $a_i$ is $\geq 2$.

Now, consider any   $\parv\in \reals^n$. By \eqref{eq:division} we have $\pi[\parv]=q[\parv]+r[\parv]$. Clearly
 $q[\parv]\in \ide G$, where $\ide G$ is here the ideal in  $\reals[\xx]$ generated by $G$. Moreover,  (b) above implies   that none of the monomials in $r[\parv]$
 is divisible by
the leading term of any polynomial in $G$. Since $G$ is a Gr\"{o}bner basis in $\reals[\xx]$,
these two properties say that $r[\parv]$  is the (unique) remainder of the
division of $\pi[\parv]$ by $G$: see e.g. \cite[Ch.2,§6,Prop.1]{Cox}.
In other words, $\pi[\parv]\bmod G = r[\parv]$.
\end{proof_of}

\vsp
\noindent
\begin{proof_of}{Lemma \ref{lemma:stab}}
We proceed by induction on $j$. The  base case $j=1$ follows from
the definition of $m$. Assuming by induction hypothesis that
$V_m=\cdots=V_{m+j}$ and that $J_m=\cdots=J_{m+j}$, we prove now
that
  $V_m=V_{m+j+1}$ and that $J_m=J_{m+j+1}$. The key to the proof
is the following fact
\vsp
\begin{eqnarray}
\pi^{(m+j+1)}[\parv]& \in &J_m  \;\; \text{ for each }\parv\in
V_m\,.\label{rel:fund}
\vsp
\end{eqnarray}
From this fact the thesis will follow, indeed:
\begin{enumerate}
\item $V_m=V_{m+j+1}$. To see this, observe that for each $\parv \in V_{m+j}=V_m$
 (the equality here follows from the induction hypothesis), it follows from  \eqref{rel:fund} that
$\pi^{(m+j+1)}[\parv]$ can be written as a finite sum of the form
$\sum_{l} h_l\cdot \pi^{(j_l)}[u_{l}]$, with $0\leq j_l\leq m$ and
$u_l\in V_m$. For each $0\leq j_l\leq m$, $\pi^{(j_l)}[u_{l}]\bmod G=0$ by assumption,
from which it easily follows that also  $\pi^{(m+j+1)}[\parv]\bmod G=
\left(\sum_{l} h_l\cdot \pi^{(j_l)}[u_{l}]\right)\bmod G=0$. This
shows that $\parv \in V_{m+j+1}$ and   proves that $V_{m+j+1}\supseteq
V_{m+j}=V_m$. The reverse inclusion is obvious;
\item $J_m=J_{m+j+1}$. As a consequence of $V_{m+j+1}=V_{m+j}(=V_m)$ (the previous point), we can write
\vsp
\begin{eqnarray*}
J_{m+j+1} & = &\ide{ \cup_{i=1}^{m+j}\pi^{(i)}[V_{m+j}] \cup
\pi^{(m+j+1)}[V_{m+j}]}\\
           & = & \ide{ J_{m+j} \cup
\pi^{(m+j+1)}[V_{m+j}]}\\
           & = & \ide{ J_{m} \cup
\pi^{(m+j+1)}[V_{m}]}
\vsp
\end{eqnarray*}
where the last step follows by induction hypothesis. From
\eqref{rel:fund}, we have that $\pi^{(m+j+1)}[V_{m}]\subseteq J_m$,
which implies the thesis for this case, as $\ide{J_m}=J_m$.
\end{enumerate}
We prove now \eqref{rel:fund}. Fix any $\parv \in V_m$. First, note that
$\pi^{(m+j+1)}[\parv]= \lie(\pi^{(m+j)}[\parv])$ (here we are using
\eqref{eq:templder}). As by induction hypothesis
$\pi^{(m+j)}[V_m]=\pi^{(m+j)}[V_{m+j}]\subseteq J_{m+j}=J_m$, we
have that $\pi^{(m+j)}[\parv]$ can be written as a finite sum $\sum_{l}
h_l\cdot \pi^{(j_l)}[u_{l}]$, with $0\leq j_l\leq m$ and $u_l\in
V_m$. Applying the rules of Lie derivatives \eqref{eq:liesum},
\eqref{eq:lieprod}, we find  that $\pi^{(m+j+1)}[\parv] =
\lie(\pi^{(m+j)}[\parv])$ equals
\vsp
\begin{equation*}
\sum_{l}\left( h_l\cdot
\pi^{(j_l+1)}[u_{l}]+ \lie(h_l)\cdot \pi^{(j_l)}[u_{l}]\right)\,.
\end{equation*}
Now, for  each $u_l$, $u_l\in V_m=V_{m+1}$,  each term
$\pi^{(j_l+1)}[u_{l}]$, with $0\leq j_l+1\leq m+1$, is by definition
in $J_{m+1}=J_m$. This proves that $\pi^{(m+j+1)}[V]\in J_m$, as
required.
\end{proof_of}

\vsp
\noindent
\begin{proof_of}{Lemma \ref{lemma:invar}}
First assume  that $\chi$ is an algebraic invariant. Take
$J=\Id(\chi)$. This is by definition an ideal. We show that $J$ is
invariant. Indeed, take any $\vv_0\in \chi$ and $p\in J$: by
hypothesis, $\xx(t;\vv_0)\in \chi$, hence $p(t;\vv_0)=0$, for each
$t\in D_{\vv_0}$, that is $p$ is a polynomial invariant for $\vv_0$.
But, by Lemma \ref{lemma:lie}, this is equivalent to
$p^{(j)}(\vv_0)=0$ for each $j\geq 0$, in particular,
$p^{(1)}(\vv_0)=0$. Since $\vv_0\in\chi$ is arbitrary, $p^{(1)}\in
J$. Since $p\in J$ is arbitrary, we have that $J$ is an invariant
ideal.

Conversely, assume that $\chi=\var(J)$ for $J$ an invariant ideal.
Of course $\chi$ is algebraic. We show that  $\chi$ is invariant,
that is that $\chi\,\impl [F]\,\chi$. Indeed, take any $\vv_0\in
\chi$ and $p\in J$: by hypothesis, $p^{(j)}\in J$ for each $j\geq
0$, hence $p^{(j)}(\vv_0)=0$ for each $j$. Again by Lemma
\ref{lemma:lie}, this means that $p(t;\vv_0)$ is identically 0.
Since $p\in J$ is arbitrary, this means that $\xx(t;\vv_0)\in\chi$
for each $t\in D_{\vv_0}$. Since $\vv_0\in\chi$ is arbitrary, we
have that $\chi$ is an invariant.
\end{proof_of}


\vsp
\noindent
\begin{proof_of}{Theorem \ref{th:instructive}}
Part (a) follows directly from Theorem \ref{th:corr}(b) and Lemma \ref{lemma:invar}.

Concerning part (b), let $\chi=\psi_\phi$,   the weakest precondition (algebraic variety) for
$\phi=\var(\pi[V])$.
  Let $(V',J')= \gdc(\chi,\pi)$. We
first prove that $\pi[V]=\pi[V']$.  Indeed, one one hand, by
definition of $I_\chi$ we have that $\pi[V]\subseteq I_\chi$ and
therefore: $\pi[V']=\pi[\reals^n]\cap I_\chi   \supseteq \pi[V]\cap
I_\chi = \pi[V]$, where the first equality comes from Theorem
\ref{th:corr}(a). On the other hand, $\psi\subseteq \chi$ by
definition of $\chi$, which implies
   $I_\chi\subseteq I_{\psi}$, therefore we have: $\pi[V]=\pi[\reals^n]\cap
I_{\psi}\supseteq \pi[\reals^n]\cap I_\chi=\pi[V']$, where the
first equality comes  again from Theorem \ref{th:corr}(a). Thus we
have proved $\pi[V]=\pi[V']$.
Now by Theorem \ref{th:corr}(b) and
   $\pi[V]=\pi[V']$, we deduce that   $J=J'$.

Now we prove that $\chi=\var(J)$. Since, by definition of $I_\chi$,
$\chi \impl[ F]\,\var(I_\chi)$, we must have
 $\chi\subseteq \var(I_\chi)$; but   $J=J'\subseteq I_\chi$ (again Theorem
\ref{th:corr}(b)), hence we have  $\var(J)=\var(J')\supseteq
\var(I_\chi)\supseteq \chi$, that is $\var(J)\supseteq \chi$. On the
other hand, by part (a),
$\var(J)$ is an
algebraic invariant, that is $\var(J)\impl[ F]\,\var(J)$; hence,
since $\pi[V]\subseteq J$ and $\var(\pi[V])\supseteq \var(J)$, we
get $\var(J)\impl[F]\,\var(\pi[V])=\phi$; the latter implies
$\var(J)\subseteq \chi$, by definition of $\chi$. In conclusion,
$\chi=\var(J)$.
\end{proof_of}

\section{Details for the collision avoidance example of Section \ref{sec:experiments}}\label{app:exp}
The following is a Gr\"{o}bner basis of the invariant ideal $J$ w.r.t. the lexicographic order induced by the following ordering of variables: $\omega_1>
 \omega_2 > x_{10}> x_{20}> y_{10}> y_{20}> d_{10}> d_{20}> e_{10}> e_{20}> x_1 > x_2 >y_1 > y_2 > d_1 > d_2 > e_1 > e_2$.
{\small
\begin{eqnarray*}
G & = \; \big{\{}  & ({x_{10}})^2 {d_{20}} + ({x_{20}})^2 {d_{20}} - 2 {x_{10}} {d_{20}} x_1 + {d_{20}} x_1^2 - 2 {x_{20}} {d_{20}} x_2 + {d_{20}} x_2^2 - 2 {x_{10}} {x_{20}} d_1 + 2 {x_{20}} x_1 d_1 + 2 {x_{10}} x_2 d_1 -  \\
&& \quad\quad  2 x_1 x_2 d_1 + ({x_{10}})^2 d_2 - ({x_{20}})^2 d_2 - 2 {x_{10}} x_1 d_2 + x_1^2 d_2 + 2 {x_{20}} x_2 d_2 - x_2^2 d_2,\\
&& ({y_{10}})^2 {e_{20}} + ({y_{20}})^2 {e_{20}} - 2 {y_{10}} {e_{20}} y_1 + {e_{20}} y_1^2 - 2 {y_{20}} {e_{20}} y_2 + {e_{20}} y_2^2 - 2 {y_{10}} {y_{20}} e_1 + 2 {y_{20}} y_1 e_1 + 2 {y_{10}} y_2 e_1 - \\
&& \quad\quad 2 y_1 y_2 e_1 + ({y_{10}})^2 e_2 - ({y_{20}})^2 e_2 - 2 {y_{10}} y_1 e_2 + y_1^2 e_2 + 2 {y_{20}} y_2 e_2 - y_2^2 e_2,\\
&& \omega_1 {x_{10}} - \omega_1 x_1 - {d_{20}} + d_2,\\
&& \omega_1 {x_{20}} - \omega_1 x_2 + {d_{10}} - d_1,\\
&& \omega_2 {y_{10}} - \omega_2 y_1 - {e_{20}} + e_2,\\
&& \omega_2 {y_{20}} - \omega_2 y_2 + {e_{10}} - e_1,\\
&& {x_{10}} {d_{10}} + {x_{20}} {d_{20}} - {d_{10}} x_1 - {d_{20}} x_2 - {x_{10}} d_1 + x_1 d_1 - {x_{20}} d_2 + x_2 d_2,\\
&& {x_{20}} {d_{10}} - {x_{10}} {d_{20}} + {d_{20}} x_1 - {d_{10}} x_2 + {x_{20}} d_1 - x_2 d_1 - {x_{10}} d_2 + x_1 d_2,\\
&& ({d_{10}})^2 + ({d_{20}})^2 - d_1^2 - d_2^2,\\
&& {y_{10}} {e_{10}} + {y_{20}} {e_{20}} - {e_{10}} y_1 - {e_{20}} y_2 - {y_{10}} e_1 + y_1 e_1 - {y_{20}} e_2 + y_2 e_2,\\
&& {y_{20}} {e_{10}} - {y_{10}} {e_{20}} + {e_{20}} y_1 - {e_{10}} y_2 + {y_{20}} e_1 - y_2 e_1 - {y_{10}} e_2 + y_1 e_2,\\
&& ({e_{10}})^2 + ({e_{20}})^2 - e_1^2 - e_2^2\;\; \big{\}}\!
\end{eqnarray*}}
Let $p$ be the first polynomial listed in $G$ above. Let us check
that $p$ is not a first integral. In fact, $p$ is not even a
Darboux polynomial. To see this, let us compute explicitly the Lie
derivative of $p$
{\small
\begin{align*}
\lie(p)=&\,\omega_1 x_{10}^2 d_1 - \omega_1 x_{20}^2 d_1 - 2 \omega_1 x_{10}
x_1 d_1 + \omega_1 x_1^2 d_1 + 2 \omega_1 x_{20} x_2 d_1 - \omega_1
x_2^2 d_1 + 2 \omega_1 x_{10} x_{20} d_2 - 2 \omega_1 x_{20} x_1 d_2
- \\
&\,2 \omega_1 x_{10} x_2 d_2 + 2 \omega_1 x_1 x_2 d_2 - 2 x_{10}
d_{20} d_1 + 2 d_{20} x_1 d_1 + 2 x_{20} d_1^2 - 2 x_2 d_1^2 - 2
x_{20} d_{20} d_2 + 2 d_{20} x_2 d_2 + 2 x_{20} d_2^2 - 2 x_2 d_2^2.
\end{align*}}\!
It is immediate to check that $\lie(p)\notin \ide{\{p\}}$  with the help of a computer algebra system. This can also be
checked manually, noting that
    $\{p\}$ is trivially a Gr\"{o}bner basis of $\ide{\{p\}}$ w.r.t. the lexicographic order, and that $\lie(p)\bmod  \{p\}=\lie(p)$. Indeed,
the leading monomial of $\lie(p)$, that is $\omega_1x_{10}^2 d_1$, is not divisible by the leading monomial of $p$, that is $x_{10}^2 d_{20}$.
 The least $m$  such that $p^{(m+1)}\in\ide{\{p,p^{(1)},...,p^{(m)}\}}$ is $m=2$.

\section{Proof of Theorem \ref{th:semialgext}, Positivstellensatz and SOS programming in Section \ref{sec:semialg}}\label{app:Psatz}
\begin{proof_of}{Theorem \ref{th:semialgext}}
By contradiction, assume there are $w_0\in X_0$ and
$w_1=\xx(t_1;w_0)\in X_U$, for some $t_1\in D_{w_0}$.
We will show
that $(w_0,w_1)\in \reals^{2N}$ is a solution of \eqref{eq:semi},
thus arriving at a contradiction. Indeed, by definition
$g_i[\xx_0/\xx](w_0,w_1)=g_i(w_0)\geq 0$ and    $h_j(w_0,w_1)=h_j(w_1)\geq 0$, for
each $i=1,...,m$ and $j=1,....,n$. Consider now the trajectory of
$\hat F$ originating from $(w_0,w_0)$, that is $\hat
\xx(t;(w_0,w_0))$: note that, by definition of $\hat F$,  $\hat
\xx(t;(w_0,w_0))=( w_0\, ,\, \xx(t;w_0))$ for each $t\in D_{w_0}$. Now,
since $\var(J)\supseteq \psi\supseteq (X_0\times X_0)\cap Id$,
 we have $(w_0,w_0)\in \var(J)$, hence, by
$\hat F$-invariance of $J$, $\hat\xx(t;(w_0,w_0))\in \var(J)$ for
each $t\in D_{(w_0,w_0)}$ (Lemma \ref{lemma:invar}). In particular,
considering $t=t_1$, we have
$\hat\xx(t_1;(w_0,w_0))=(w_0\,,\,\xx(t_1;w_0))=(w_0,w_1)  \in \var(J)$.
But this means $q_i(w_0,w_1)=0$ for $i=1,...,k$. In conclusion,
$(w_0,w_1)$ is a solution of \eqref{eq:semi}.
\end{proof_of}

\vsp
When working in algebraically closed fields, like $\cplx$,
 Hilbert's Nullstellensatz \cite{Cox} implies that a system of polynomial
 equations $P$  has no solution if and only if
 $1\in \sqrt{\ide{P}}$. This gives a simple criterion to check
 if $P$ is solvable.
The following result,  often considered as the real algebraic counterpart
 of Hilbert's Nullstellensatz, is due
to Stengle \cite{Stengle}. Let us  introduce the necessary terminology. In what follows,
all polynomials
 are in  $\reals[\xx]$ for some fixed $\xx=(x_1,...,x_N)$.
A polynomial $s$
is a \emph{Sum of Squares (SOS)} if $s=\sum_j h^2_j$, for some polynomials $h_j$. Given a finite set
of polynomials $A=\{f_1,...,f_n\}$, the \emph{cone} generated by $A$ is
$C(A)\defi \{\sum_{\sigma\subseteq\{1,...,n\}}s_\sigma
\Pi_{j\in \sigma}f_j\,:\,s_\sigma \text{ are   SOS }\}$.

\begin{theorem}[Positivstellensatz]\label{th:Psatz}
Let $A=\{f_1,...,f_n\}$ and $B=\{g_1,...,g_m\}$ be two sets of polynomials.
 The system of (in)equations
$\{f_1\geq 0,...,f_n\geq 0,\, g_1=0,...,g_m=0\}$ has no solution in $\reals^N$
 if and only if there are $f\in C(A)$ and
$g\in \ide{B}$ such that $f+g+1=0$.
\end{theorem}

If one writes $f\in C(A)$ as
$f=s_\es+\sum_{\es\neq\sigma\subseteq\{1,...,n\}} s_\sigma\Pi_{j\in \sigma}f_j$
and $g\in \ide{B}$ as $g=\sum_{i=1}^m h_ig_i$, then finding $f\in C(A)$ and $g\in \ide{B}$ such that
$f+g+1=0$   can be formulated as follows:
\begin{eqnarray}\label{pro:psatz}
\text{Find polynomials  $h_i$'s and SOS $s_\sigma$'s such that }
-\left(\sum_{\es\neq\sigma\subseteq\{1,...,n\}}\!\!\!s_\sigma\Pi_{j\in \sigma}f_j+\sum_{i=1}^m h_ig_i+1\right)    \text{ is SOS.}
\end{eqnarray}
Now a polynomial $s$ is SOS if and only if there is a vector of monomials
 $Z=(\alpha_1,...,\alpha_K)$ and a real symmetric
positive semidefinite $K\times K$ matrix $M$ such that $s=ZMZ^T$.
Once   bases of monomials have been fixed for each of the (unknown) polynomials
$s_\sigma$ and $h_i$,
one can   consider a relaxation of problem \eqref{pro:psatz}, whereby one searches for
polynomials   built from those bases satisfying \eqref{pro:psatz} (empty bases are allowed).
Problem \eqref{pro:psatz} becomes in this way a \emph{semidefinite programming problem} \cite{Parrilo},
with one variable for each  (unknown)
polynomial  coefficient, and
 constraints   given by the various positive semidefinite conditions and by the equation \eqref{pro:psatz}.
 In fact, there are   tools,
like SOSTOOLS \cite{SOSTOOLS}, to efficiently convert relaxations of
problem \eqref{pro:psatz} into a semidefinite programming problem
and then try to solve it via numerical techniques. If successful, the obtained SOS polynomial is a certificate of insolvability
of the original system. Although the number of terms
 in \eqref{pro:psatz} is potentially exponential in $n$,
experience has shown that, in practice,  this technique tends
to yield short (low degree) certificates, if the original
system is actually not solvable.

For the continuous semialgebraic systems in Examples \ref{ex:3LV} and \ref{ex:springmass},
denoting by $p_0, p_U$
the polynomials defining $X_0$ and $X_U$ in each case,
problem \eqref{pro:psatz} takes  the following concrete form
\begin{eqnarray}\label{pro:psatz1}
\text{Find   $h_1,h_2$  and SOS $s_1,s_2,s_3$ such that }
-\left( s_1p_{0}+s_2p_U+s_3p_0p_U + h_1q_1+h_2q_2 +1
 \right)    \text{ is SOS.}
\end{eqnarray}
For Example \ref{ex:3LV}, fixing a maximum degree of $1$ for the $h_i$'s and of
 $2$ for the $s_\sigma$'s and running
  SOSTOOLS under Matlab\footnote{Matlab scripts for both examples are available at \url{https://github.com/micheleatunifi/postconditions/blob/master/SOSSafety.m}.}
  we solve \eqref{pro:psatz1}, finding
   the following polynomials, which yield a low-degree certificate:
{\small
\begin{eqnarray*}
s_1 & = &
  2.8733 x^2 - 0.15408 x x_0 - 0.6222 x y - 0.11367 x y_0 - 0.81927 x z
  + 4.2254 x_0^2 - 1.4885 x_0 y - 2.4633 x_0 y_0 - \\
  && 0.61469 x_0 z + 1.5901 y^2
  - 1.3772 y y_0 - 1.4399 y z + 4.071 y_0^2 - 0.58541 y_0 z + 3.1195 z^2\\
s_2 & = & 1.7131 x^2 + 0.063147 x x_0 - 0.33468 x y + 0.074991 x y0 + 0.17391 x z
  + 1.2985 x_0^2 - 0.51261 x_0 y + \\
  && 0.19351 x_0 y_0 + 0.15971 x_0 z + 0.5814
   y^2 - 0.53211 y y_0 - 0.48663 y z + 1.3512 y_0^2 + 0.17931 y_0 z + 1.9153
   z^2\\
s_3 & = & 0\\
h_1 &= & 0.49855 x + 0.21264 x_0 - 0.29621 y + 0.18602 y_0 + 0.35564 z + 14.8214\\
h_2 & = & 13.096 x + 9.6921 x_0 + 0.013631 y - 36.2677 y_0 + 22.7601 z - 16.3507\,.
\end{eqnarray*}
}%
This takes about 0.4s on a Core i5 machine under Windows 10.
 A certificate for Example \ref{ex:springmass} is found
in a similar way; we omit here the lengthy description of the corresponding polynomials.

\begin{remark}[dealing with roundoff errors]\label{rem:SOSerr}{\em
It is   important to ensure that the   SOS  decomposition
found numerically   actually corresponds to a true solution, and it is not the result
of   roundoff errors that may arise when working in floating point arithmetic.  There are several ways of doing this: for instance, by computing exact rational solutions, that   can be fully
verified symbolically, or by tweaking the numerical coefficients of  a candidate   solution; see for instance \cite{ParriloNum,ChineseSOS}. In particular, the SOSTOOLS  \textsf{\texttt{findSOS}}  procedure incorporates an  experimental \textsf{\texttt{'rational'}}  option, that   will try  to produce an exact rational SOS
representation of the input polynomial. Using this option, we have checked that the solution found for the Lotka-Volterra example is indeed SOS. On the other hand, the \textsf{\texttt{'rational'}} option fails to find a rational representation in the spring-mass example.
}
\end{remark}
\fi
\end{document}